\documentclass[final,onefignum,onetabnum,reqno]{amsart}

\usepackage[colorlinks=true,allcolors=blue]{hyperref}
\hypersetup{hypertexnames=false}

\usepackage[foot]{amsaddr}
\usepackage{amssymb,amsfonts,latexsym,amsmath}
\usepackage{amsthm}
\usepackage{thmtools}
\usepackage[capitalise]{cleveref}

\usepackage[square,numbers]{natbib}
\usepackage{bm}
\usepackage{bbold}

\usepackage{mathtools}

\crefname{equation}{}{}

\newtheorem{theorem}{Theorem}[section]
\newtheorem{lemma}[theorem]{Lemma}
\crefname{lemma}{Lemma}{Lemmas}
\Crefname{lemma}{Lemma}{Lemmas}
\newtheorem{corollary}[theorem]{Corollary}
\crefname{corollary}{Corollary}{Corollaries}
\Crefname{corollary}{Corollary}{Corollaries}
\newtheorem{proposition}[theorem]{Proposition}
\crefname{proposition}{Proposition}{Propositions}
\Crefname{proposition}{Proposition}{Propositions}

\theoremstyle{remark}
\newtheorem{remark}{Remark}
\crefname{remark}{Remark}{Remarks}
\Crefname{remark}{Remark}{Remarks}

\renewcommand{\le}{\leqslant}
\renewcommand{\ge}{\geqslant}

\let\phi=\varphi 
\renewcommand{\epsilon}{\varepsilon}
\undef\ul
\newcommand{\ol}[1]{\overline{#1}}
\newcommand{\ul}[1]{\underline{#1}}
\newcommand{\wt}[1]{\widetilde{#1}}
\newcommand{\wh}[1]{\widehat{#1}}

\DeclareMathOperator{\dist}{dist}

\undef\Span
\DeclareMathOperator{\Span}{Span}
\undef\div
\DeclareMathOperator{\div}{div}
\DeclareMathOperator{\supp}{supp}

\DeclareMathOperator{\diam}{diam}
\DeclareMathOperator{\Tr}{Tr}

\newcommand{\RR}{\mathbb{R}}
\newcommand{\CC}{\mathbb{C}}
\newcommand{\ZZ}{\mathbb{Z}}
\newcommand{\NN}{\mathbb{N}}
\newcommand{\PP}{\mathbb{P}}

\newcommand{\vGamma}{{\bm{\Gamma}}} 

\newcommand{\SO}{\mathrm{SO}}
\newcommand{\per}{\mathrm{per}} 
\newcommand{\cl}{\mathrm{cl}} 
\newcommand{\LL}{\mathrm{LL}} 
\newcommand{\xc}{\mathrm{xc}} 
\newcommand{\TF}{\mathrm{TF}} 
\newcommand{\LO}{\mathrm{LO}} 
\newcommand{\UEG}{\mathrm{UEG}} 
\newcommand{\NUEG}{\mathrm{NUEG}} 
\newcommand{\GS}{\mathrm{GS}} 
\newcommand{\Mo}{\mathrm{Mo}} 
\newcommand{\SCE}{\mathrm{SCE}} 

\newcommand{\wconv}{\rightharpoonup} 
\newcommand{\iden}{\mathbb{1}} 
\newcommand{\lapl}{\bm{\Delta}} 
\newcommand{\grad}{\bm{\nabla}} 
\newcommand{\tetra}{\mathbb{\Delta}} 
\newcommand{\dd}{\mathrm{d}} 

\newcommand{\Lloc}{L_{\mathrm{loc}}} 
\newcommand{\Hloc}{H_{\mathrm{loc}}} 
\newcommand{\dm}{\mathcal{D}} 

\def\Xint#1{\mathchoice
{\XXint\displaystyle\textstyle{#1}}%
{\XXint\textstyle\scriptstyle{#1}}%
{\XXint\scriptstyle\scriptscriptstyle{#1}}%
{\XXint\scriptscriptstyle\scriptscriptstyle{#1}}%
\!\int}
\def\XXint#1#2#3{{\setbox0=\hbox{$#1{#2#3}{\int}$ }
\vcenter{\hbox{$#2#3$ }}\kern-.6\wd0}}

\def\dashint{\Xint-}

\newcommand{\TC}{\mathcal{T}}
\newcommand{\EC}{\mathcal{E}}
\newcommand{\Cc}{\mathcal{C}}
\newcommand{\LC}{\mathcal{L}}
\newcommand{\OC}{\mathcal{O}}

\newcommand{\SF}{\mathfrak{S}}

\title[Non-uniform Electron Gas]{The non-uniform electron gas}

\author{Mih\'aly A. Csirik$^{*,1}$
 \and Andre Laestadius$^{1,2}$}
\address[*]{Corresponding author, email: \url{csirik@gmail.com}}
\address[1]{Department of Computer Science, Oslo Metropolitan University, Norway}
\address[2]{Hylleraas Centre for Quantum Molecular Sciences, Department of Chemistry, University of Oslo, Norway}

\begin{document}

\begin{abstract}
The non-uniform (or inhomogeneous) electron gas has received much attention in many-body quantum mechanics and quantum chemistry in the early days of density functional theory, mainly as a theoretical device to construct gradient approximations via linear response theory. 
In this article, motivated by the recent works of Lewin, Lieb and Seiringer, we propose a definition of the quantum (resp. classical) non-uniform electron gas through the use of the grand-canonical Levy--Lieb functional (resp. the grand-canonical strictly correlated electrons functional), establish these systems as rigorous thermodynamic limits and analyze their basic properties.
The non-uniformity of the gas comes from an arbitrary lattice-periodic background density.
\end{abstract}

\maketitle

\section{Introduction}

Before the widespread acceptance of the name ``density functional theory'' early researchers of the field referred to what we call DFT today as 
the theory of the ``inhomogeneous electron gas'' \cite{hohenberg1964inhomogeneous,lundqvist1983theory}.
As far as we can tell, the (weakly) inhomogeneous electron gas is regarded a ``fictitious system'', or a stepping stone, and not a ``real'' physical system like the uniform electron gas (UEG).
Furthermore, it is always understood in a perturbative sense, i.e. through its linear response \cite{giuliani2008quantum,engel2011density}. These density perturbations are measured relative to the constant density of the homogeneous gas. In \cite{engel2011density} ``both the case of a localized perturbation as well as some periodic structure'' are claimed to be covered. The treatment of a localized perturbation in the thermodynamic limit remains somewhat unclear. Hence, we will consider the case of a lattice-periodic inhomogeneity in this work.
Rather than relying on linear response theory, we treat the inhomogeneous electron gas as a genuinely nonlinear (i.e. non‑perturbative) system. The weakly inhomogeneous electron gas may be considered within this nonlinear framework as a perturbation of the homogeneous gas by a slowly varying inhomogeneity.

Henceforth, we shall adopt the term \emph{non-uniform electron gas} (NUEG) in order to better align with the existing mathematical physics literature pertaining to the UEG~\cite{lewin2018statistical}. The main reason for this is that we will also employ density functionals to define the NUEG energy. The alternative path would be to consider the thermodynamic limit of an electronic ground-state problem like in the case of the Jellium~\cite{lieb1975thermodynamic}. In the Jellium model the external potential is obviously constant. However, it is unclear what external potential (or whether such potential exists) would generate a prescribed lattice-periodic inhomogeneity as a ground-state density. This issue is related to the so-called $v$-representability problem, which is poorly understood at the present---there is no analytic description of the $v$-representability set except in 1D on a finite interval with Neumann boundary conditions~\cite{Corso2026,corso-lae2025}.
The approach of \cite{lewin2018statistical} elegantly avoids these complications by defining the UEG energy through the use of density functionals and demanding the density to be constant. 
The price to pay is that we lose the connection with the ground-state problem and have to deal with the constrained optimization problem, which may not even have the usual Schr\"odinger equation as its Euler--Lagrange equation. The said connection is still unresolved for the quantum UEG, but is settled for the classical case \cite{cotar2017equality,lewin2019floating}.

Our definition of the NUEG in the 3D Coulomb case involves a fixed lattice-periodic \emph{inhomogeneity} $\zeta:\RR^3\to\RR_+$ having natural integrability and regularity properties. 
This $\zeta$ replaces the constant density $\rho_0$ in the
definitions of \cite{lewin2018statistical} based on the strictly-correlated electrons (SCE) functional (classical case) and the grand-canonical Levy--Lieb functional \cite{lewin2019universal}
(quantum case). In order to obtain a well-behaved energy per volume, we average the energy over the translations and rotations of $\zeta$. More precisely,
the (3D Coulomb) indirect energy per volume of the classical NUEG with inhomogeneity $\zeta$ in a bounded domain $\Omega\subset\RR^3$ is taken to be
\begin{multline*}
e_\Omega^\cl(\zeta)=\frac{1}{|\Omega|}\int_{\SO(3)}\dd R \dashint_{\RR^3}\dd a\,\Bigg[ \inf_{\substack{\PP\,\text{GC prob.}\\\rho_\PP=\iden_\Omega \zeta(R(\cdot-a))}} 
\sum_{n\ge 2} \sum_{1\le j<k\le n} \int_{\RR^{3n}} \frac{\PP_n(\dd x_1,\ldots,\dd x_n)}{|x_j-x_k|}\\
-\frac{1}{2}\iint_{\RR^3\times\RR^3} \frac{\iden_\Omega(x) \zeta(R(x-a))\iden_\Omega(y) \zeta(R(y-a))}{|x-y|}\, \dd x\dd y \Bigg]
\end{multline*}
in the classical case and analogously in the quantum case, which we denote by $e_{\Omega,\eta_\delta}^\hbar(\zeta)$. Here, a regularization function $\eta_\delta$ must be used, because quantum densities cannot be cut off sharply.
 Noting that the $a$-integrand above is lattice-periodic and locally integrable, the mean value over $\RR^3$ exists and equals the mean value over 
a unit cell of the lattice. This is a ``floating crystal-like'' construction, where
the lattice is allowed to float in the container $\Omega$.
Under the sole assumption that $\zeta^{4/3}$ is locally integrable 
(and that $\sqrt{\zeta}$ is in $\Hloc^1$ in the quantum case) we can prove that the thermodynamic limit of $e_\Omega^\cl(\zeta)$ (resp. $e_{\Omega,\eta_\delta}^\hbar(\zeta)$)
exists for rather general domain sequences. We denote the limiting functional by $e_\NUEG^\cl(\zeta)$ (resp. $e_\NUEG^\hbar(\zeta)$). In the quantum case, we also offer a slightly different definition based on a transition region near the boundary where the density is allowed to ``relax'' in an arbitrary manner analogously to \cite{lewin2018statistical}. 
We prove the equivalence of the two definitions.

Clearly, by putting $\zeta\equiv\rho_0$ we recover the UEG energy exactly, even at finite volume. 
We consider the slowly varying limits $e_\NUEG^{\cl/\hbar}(\zeta(\lambda\cdot))$ as $\lambda\to 0$ and establish the rate of its convergence to the UEG energy 
dictated by the local density approximation (LDA)
both in the classical-, and the quantum case. These follow from an adaptation of the method of \cite{lewin2020local}. 

In addition to establishing the NUEG as a physically reasonable system with expected properties, our work provides a complementary perspective on DFT---in fact, instead of integrable \emph{densities}, the functionals $e_\NUEG^{\cl/\hbar}(\zeta)$ introduced involve locally integrable \emph{inhomogeneities} $\zeta$ that are \emph{not} globally integrable. It would be interesting to extend the allowed set of inhomogeneities to a broader class, such as to some space of almost periodic functions.
On the technical side, we improve a very useful spatial decoupling estimate in \cref{impdecup} that appeared first in \cite{lewin2020local}, which is a crucial tool in the quantum case. 
\vspace{.3em}\\
\noindent\emph{Organization of the paper.} In \cref{sec:Prel}, we recall the necessary definitions and facts for our discussion. In \cref{sec:MR}, we present our main results in detail. \cref{clnuegsec} contains the definition of the classical non-uniform electron gas through its thermodynamic limit (\cref{clnuegthm}).
We also establish the local density approximation for slowly varying inhomogeneities (\cref{cllda}). The proofs of these results are provided in \cref{clproofs}.
The quantum case is covered in \cref{qnuegsec}, our main results being \cref{thermolimthm} and \cref{qlda}, which are analogs of the classical case. The proofs for the quantum gas may be found in \cref{qproofs}.
Finally, the Appendices contain the proofs of some auxiliary results. 
\vspace{.3em}\\
\noindent\emph{Acknowledgments.}
The authors are grateful to Erik I. Tellgren for many helpful discussions.
A.L. and M.A.Cs. have received funding from the ERC-2021-STG
under grant agreement No. 101041487 REGAL. A.L. was also funded by the Research
Council of Norway through CoE Hylleraas Centre for Quantum Molecular Sciences
Grant No. 262695. 

\section{Preliminary notions}\label{sec:Prel}

\subsection{Notations}

Throughout the paper, we set $C_1=(-\frac{1}{2},\frac{1}{2})^d$ for the origin-centered $d$-dimensional open unit cube. Its scalings will be denoted as $C_\ell=\ell C_1$ for any side length $\ell>0$. 
We will make use of the mean value notation
$$
\dashint_A u=\frac{1}{|A|}\int_A u
$$
for any $A\subset\RR^d$ of finite measure and $u: \RR^d\to\CC$. 
The group of orientation-preserving rotations of $\RR^d$ is denoted by $\SO(d)$. The group of all permutations of $\{1,\ldots,N\}$ is denoted by $\SF_N$. 

\subsection{Lattice-periodic functions}
Let $\LC\subset\RR^d$ be a lattice, i.e. is a discrete additive subgroup of $\RR^d$.
A unit cell $\Lambda\subset\RR^d$ of the lattice $\LC$ is a bounded open set such that its lattice translates $(\Lambda+l)_{l\in\LC}$ form a tiling of $\RR^d$, i.e. a pairwise disjoint family of sets such that $\bigcup_{l\in\LC} \ol{(\Lambda+l)}=\RR^d$.
A function $u:\RR^d\to\CC$ is said to be $\LC$-periodic if $u(x+l)=u(x)$ for all $l\in\LC$ and almost all $x\in\RR^d$. Such functions are almost everywhere determined by their values in a unit cell $\Lambda$.

We will repeatedly make use of the following elementary fact, the proof of which is provided in \cref{furthersec} for the convenience of the Reader.
\begin{theorem}\label{permeanthm}
For any $\LC$-periodic complex-valued function $u\in\Lloc^1(\RR^d)$, the mean value of $u$ over $\RR^d$ exists and equals its mean value over a unit cell $\Lambda\subset\RR^d$ of $\LC$, i.e.
\begin{equation}\label{permean}
\dashint_{\RR^d} u:=\lim_{L\to\infty} \dashint_{C_L+a} u =\dashint_{\Lambda} u,
\end{equation}
independently and uniformly in $a\in\RR^d$.
\end{theorem}

For a fixed lattice $\LC$, we will adopt the notation $L_\per^p(\LC)$ for the space of $\LC$-periodic functions such that $u\in \Lloc^p(\RR^d)$ (equivalently, $u\in L^p(\Lambda)$ for a unit cell $\Lambda$ of $\LC$). The continuous inclusions $L^\infty_\per(\LC)\subset L^p_\per(\LC)\subset L^q_\per(\LC)\subset L^1_\per(\LC)$ hold true for all $1<q\le p<\infty$. Similarly, we introduce the Sobolev space $H^1_\per(\LC)$.

\subsection{Classical density functional theory}
Following \cite{di2025grand}, we say that $\PP=(\PP_n)_{n\ge 0}$ is a \emph{grand-canonical probability} if 
$$
\PP_0 + \sum_{n\ge 1} \PP_n(\RR^{dn})=1,
$$
where $\PP_0\in[0,1]$ and $\PP_n$ is a finite positive symmetric Borel measure on $(\RR^d)^n$ for all $n\ge 1$. The symmetry of $\PP_n$ means that 
$$
\PP_n(B_{\sigma(1)}\times\ldots \times B_{\sigma(n)})=\PP_n(B_1\times\ldots\times B_n)
$$
for every permutation $\sigma\in\SF_n$.
The \emph{density} of $\PP$ is then the marginal
$$
\rho_\PP(B)=\PP_1(B) + \sum_{n\ge 2} n\PP_n(B,\RR^{d(n-1)}).
$$
Suppose that $0<s<d$. The Riesz energy of a grand-canonical probability $\PP$ is given by 
$$
\Cc_s(\PP)=\sum_{n\ge 2}\int_{\RR^{dn}} \sum_{1\le j<k\le n}\frac{1}{|x_j-x_k|^s}\, \dd\PP_n(x_1,\ldots,x_n).
$$
Next, for any $\rho\in L^1(\RR^d)\cap L^{\frac{2d}{2d-s}}(\RR^d)$ we introduce the (classical) indirect energy 
$$
\boxed{E^\cl(\rho)=F_\SCE(\rho)-D_s(\rho)}
$$
where the \emph{(grand-canonical) strictly correlated electrons functional} $F_\SCE(\rho)$ is defined
as
$$
F_\SCE(\rho)=\inf_{\substack{\PP\,\text{GC prob.}\\\rho_\PP=\rho}} \Cc_s(\PP),
$$
and direct term is given by the positive quadratic form
$$
D_s(\rho)=\frac{1}{2}\iint_{\RR^{2d}} \frac{\rho(x)\rho(y)}{|x-y|^s}\,\dd x\dd y<\infty.
$$
In the 3D Coulomb case we will drop the subscripts from $\Cc_s$ and $D_s$ and simply write $\Cc$ and $D$. 
It is immediate that $E^\cl(\rho)$ is isometry-invariant: $E^\cl(\rho(R(\cdot-a)))=E^\cl(\rho)$ for every $R\in SO(d)$ and $a\in\RR^d$. 
The Lieb--Oxford inequality and a simple trial state construction involving pure tensor products give the bounds
\begin{equation}\label{Eclapriori}
-c_\LO(d,s) \int_{\RR^d} \rho^{1+s/d}\le  E^\cl(\rho)\le 0,
\end{equation}
where $c_\LO(d,s)>0$ is a universal constant that depends only on $d$ and $s$. 
Henceforth, we shall work with densities $\rho\in L^1(\RR^d)\cap L^{1+s/d}(\RR^d)$, which automatically have $\rho\in L^{\frac{2d}{2d-s}}(\RR^d)$ as required above. 
Aside from the preceding \emph{a priori} bound, the most important property of the indirect energy is its subadditivity
\begin{equation}\label{Eclsubadd}
E^\cl(\rho_1+\rho_2)\le E^\cl(\rho_1) + E^\cl(\rho_2)
\end{equation}
for every $\rho_1,\rho_2\in L^1(\RR^d)\cap L^{1+s/d}(\RR^d)$. Notice that the densities $\rho_1$ and $\rho_2$ may overlap. 
Another useful property of the indirect energy is its simple scaling law,
\begin{equation}\label{Eclscale}
E^\cl(\lambda\rho(\lambda^{1/d}\cdot)) = \lambda^{s/d} E^\cl(\rho)
\end{equation}
for every $\lambda>0$. In the quantum case, neither \cref{Eclsubadd} nor \cref{Eclscale} holds. 

\subsection{Grand-canonical Levy--Lieb functional}
In this section, we introduce notations and recall basic facts about the quantum version of the above functionals. Throughout the paper when discussing the quantum case, we shall restrict ourselves to the 3D Coulomb case.
Also, we only consider spinless electrons for simplicity of notations. 

Let $\dm$ denote the set of fermionic Fock space states on $L^2(\RR^3)$ with finite kinetic energy having the additional property that they commute with the number operator. For a state $\vGamma=(\Gamma_n)_{n\ge 0}$ in $\dm$ we denote its one-particle reduced density matrix by
$$
\gamma_\vGamma(x,y)=\sum_{n\ge 1} n \int_{\RR^{3(n-1)}} \Gamma_n(x,x_2,\ldots,x_n;y,x_2,\ldots,x_n)\,\dd x_2\ldots \dd x_n.
$$
Further, we denote the density of $\vGamma$ by $\rho_\vGamma(x)=\gamma_\vGamma(x,x)$. This
has $\rho_\vGamma\in L^1(\RR^3,\RR_+)$ and 
by the Hoffmann--Ostenhof inequality $\grad\sqrt{\rho_\vGamma}\in L^2(\RR^3)$. 

For any $\vGamma\in\dm$ let us introduce the notations
\begin{equation*}
\begin{aligned}
\TC(\vGamma)
&= \Tr_{L^2(\RR^3)} (-\lapl \gamma_\vGamma), \\ 
\Cc(\vGamma) &=\sum_{n\ge 2}\sum_{1\le j<k\le n} \Tr_{L^2(\RR^{3n})}\left( \frac{1}{|x_j-x_k|} \Gamma_n \right)
\end{aligned}
\end{equation*}
for the \emph{kinetic-}, and the \emph{Coulomb energy of the state} $\vGamma$. Then the \emph{grand-canonical Levy--Lieb functional} is defined as
$$
F_\LL^\hbar(\rho)=\inf_{\substack{\vGamma\in\dm\\\rho_\vGamma=\rho}} \Big[ \hbar^2\TC(\vGamma) + \Cc(\vGamma) \Big]
$$
for every $\rho\in L^1(\RR^3,\RR_+)$ such that $\grad\sqrt{\rho}\in L^2(\RR^3)$, and $+\infty$ otherwise. By dropping the Coulomb interaction, we define the \emph{mixed-state kinetic energy functional} $T^\hbar(\rho)$.
The infimum in $F_\LL^\hbar(\rho)$ (resp. $T^\hbar(\rho)$) is attained \cite{lewin2019universal}. Moreover, for every fixed $\rho$, the function $\hbar\mapsto F_\LL^{\sqrt{\hbar}}(\rho)$ is a minimum of affine functions (in $\hbar$) of positive slope, hence it is concave increasing.

Using $F_\LL^\hbar(\rho)$, we can define the \emph{(quantum) indirect energy} 
$$
\boxed{E^\hbar(\rho)=F_\LL^\hbar(\rho)-D(\rho)}
$$
Notice that contrary to the classical indirect energy $E^\cl(\rho)$, the quantum variant may be positive. However, it can be shown that the \emph{exchange-correlation energy} 
$$
E_\xc^\hbar(\rho)=E^\hbar(\rho)-T^\hbar(\rho)
$$
is nonpositive. Various \emph{a priori} bounds on $E^\hbar(\rho)$ are recalled in \cref{qknownbd} below.

There is no known simple scaling relation for $E^\hbar(\rho)$ analogous to \cref{Eclscale}, because the kinetic energy scales differently from the Coulomb energy. However, if we allow $\hbar$ to change, the following obvious fact holds
\begin{equation}\label{Ehbarscale}
E^\hbar(\lambda\rho(\lambda^{1/3}\cdot))=\lambda^{1/3}E^{\lambda^{1/6}\hbar}(\rho),
\end{equation}
for all $\lambda>0$.
As mentioned above, $E^\hbar(\rho)$ is not subadditive. A substitute for the subadditivity property 
pioneered by \cite{lewin2020local} is discussed in \cref{impdecoupsec} below.

\section{Main results}\label{sec:MR}

\subsection{Classical non-uniform electron gas}\label{clnuegsec}

Let $0<s<d$ and fix a \emph{(classical) inhomogeneity} $\zeta\in L_\per^{1+s/d}(\LC)$. Setting $\rho=\iden_\Omega\zeta(\cdot-a)$ in $E^\cl(\rho)$, the bounds \cref{Eclapriori} imply that the function $a\mapsto |\Omega|^{-1} E^\cl(\iden_\Omega\zeta(\cdot-a))$ 
 is in $L_\per^1(\LC)$, hence its mean value over $\RR^d$ is well-defined and finite (see \cref{permeanthm}).
  Further, replacing $\zeta$ by $\zeta(R\cdot)$ and averaging over
 $R\in\SO(d)$, we arrive at our definition\footnote{Our definition is vaguely motivated by the abstract setting of \cite{HAINZL2009454}, where only an average translation invariance is assumed (see Assumption (A3), ibid.) about the energy.} of the \emph{classical indirect energy per volume}
$$
\boxed{e_\Omega^\cl(\zeta)=\int_{\SO(d)}\dd R \dashint_{\RR^d}\dd a\, \frac{E^\cl\bigl(\iden_\Omega \zeta(R(\cdot-a))\bigr)}{|\Omega|}}
$$
Clearly, $e_\Omega^\cl(\zeta)$ is isometry-invariant with respect to both $\Omega$ and $\zeta$ (separately) and verifies the \emph{a priori} bound
\begin{equation}\label{ecllobound}
-c_\LO(d,s) \dashint_{\RR^d} \zeta^{1+s/d}\le e_\Omega^\cl(\zeta)\le 0. 
\end{equation}
Physically speaking, our definition describes the energy of a ``floating crystal'' where the lattice-periodic inhomogeneity $\zeta$ 
is allowed to translate and rotate freely in a ``container'' $\Omega$, so as to neutralize the energy fluctuations coming from the boundary layer of lattice cells. 
Notice that the definition is lattice-independent in the sense that it is meaningful for every lattice-periodic inhomogeneity, and the definition remains exactly the same.
This is because only the mean value over $\RR^d$ is used, and as long as the integrand is lattice-periodic with respect to \emph{some} lattice, the mean value makes sense.

The isometry-invariance of the indirect energy per volume allows us to deduce the existence of its thermodynamic limit using standard methods.
To state our result, we recall a geometric condition for proving thermodynamic limits. Let $\Omega\subset\RR^d$ be a bounded open and connected set. Let $\kappa:[0,t_0)\to\RR_+$ be such that $\kappa(t)\to 0$ as $t\to 0+$. 
The set $\Omega$ is said to have \emph{$\kappa$-regular boundary (in the sense of Fisher \cite{fisher1964free})}, if for every $0\le t<t_0$
$$
|\{x\in\RR^d : \dist(x,\partial\Omega)\le |\Omega|^{1/d}t\}|\le \kappa(t)|\Omega|.
$$
A sequence of bounded domains $\{\Omega_N\}\subset\RR^d$ is said have a uniformly $\kappa$-regular boundary, if $\Omega_N$ is has a $\kappa$-regular boundary with the same $\kappa$ for all $N$.

\begin{theorem}[Classical non-uniform electron gas]\label{clnuegthm} 
Fix an inhomogeneity $\zeta\in L_\per^{1+s/d}(\LC)$.
Let $\{\Omega_N\}\subset\RR^d$ be a sequence of bounded domains having uniformly $\kappa$-regular boundary and $|\Omega_N|\to\infty$.
Then the following thermodynamic limit exists
\begin{equation}\label{clnueglim}
\lim_{N\to\infty} e_{\Omega_N}^\cl(\zeta)=e_\NUEG^\cl(\zeta)
\end{equation}
and is independent of the sequence $\{\Omega_N\}$.
\end{theorem}

In the 3D Coulomb case, we can also derive a convergence rate estimate for dilated tetrahedra, see \cref{clconvrate}. Before proceeding, we make a few remarks.

\begin{remark}
\noindent (i) Using the scaling relation $E^\cl(\lambda\rho(\lambda^{1/d}\cdot))=\lambda^{s/d} E^\cl(\rho)$ we have that
$$
e_\NUEG^\cl(\lambda \zeta(\lambda^{1/d}\cdot))=\lambda^{1+s/d} e_\NUEG^\cl(\zeta).
$$
The first observation we make is that it is enough to know $e_\NUEG^\cl(\zeta)$ for $\zeta$ having $\dashint\zeta=1$ . In fact, we may write $\zeta=\rho_0\zeta_0(\rho_0^{1/d}\cdot)$, where $\dashint\zeta_0=\dashint \zeta_0(\rho_0^{1/d}\cdot)=1$ and $\rho_0=\dashint\zeta$ are uniquely determined,
so that
$$
e_\NUEG^\cl(\zeta)=\rho_0^{1+s/d} e_\NUEG^\cl(\zeta_0).
$$
Roughly speaking, $e_\NUEG^\cl(\zeta)$ depends on $\dashint\zeta$ in the usual, expected manner. However, it depends on the ``shape'' of $\zeta$ in a nontrivial way. 
\vspace{.3em}\\
\noindent (ii) Without loss we can assume that the lattice $\LC$ has a unit cell $\Lambda\subset\RR^d$ with normalized volume $|\Lambda|=1$. In fact, we can replace an $\LC$-periodic $\zeta$ such that $\dashint\zeta=1$ by $\zeta(|\Lambda|^{-1/d}\cdot)$ which is $(|\Lambda|^{-1/d}\LC)$-periodic, and also has mean value 1. 
\vspace{.3em}\\
\noindent (iii)
As expected, when $\zeta\equiv \rho_0$ is constant, then $\zeta_0\equiv 1$ and we recover the uniform electron gas energy 
$e_\NUEG^\cl(\rho_0)=c_\UEG\rho_0^{1+s/d}$,
where the negative constant $c_\UEG=e_\NUEG^\cl(1)$ is the classical UEG energy as defined in \cite{lewin2018statistical}. Lieb and Narnhofer \cite{lieb1975thermodynamic} gave an estimate on the Jellium ground-state energy in 3D Coulomb case, which implies
$c_\UEG \ge - \frac{3}{5} \left(\frac{9\pi}{3}\right)^{1/3}\approx -1.4508$.
\vspace{.3em}\\
\noindent (iv)
It is clear that the real function $\rho_0\mapsto c_\UEG\rho_0^{1+s/d}$ is subadditive since it is concave and vanishes at $\rho_0=0$. More generally, the subadditivity property of $E^\cl(\rho)$ is inherited by the functional $\zeta\mapsto e_\NUEG^\cl(\zeta)$, i.e. we have 
$$
e_\NUEG^\cl(\zeta+\zeta')\le e_\NUEG^\cl(\zeta)+e_\NUEG^\cl(\zeta')
$$
for all  $\zeta,\zeta'\in L_\per^{1+s/d}(\LC)$.
\end{remark}

Next, we show how the NUEG-, and UEG energies are related when the inhomogeneity is slowly varying. 
Here, we restrict ourselves to the 3D Coulomb case.

\begin{theorem}[Local density approximation]\label{cllda}
Suppose that $d=3$ and $s=1$. Let $p>3$ and $0<\theta<1$ such that $\theta p\ge 4/3$. For every $\epsilon>0$ and every inhomogeneity $\zeta$ such that $\grad\zeta^\theta\in\Lloc^p(\RR^3)$ the bound
\begin{equation}\label{clldaest}
\left| e_\NUEG^\cl(\zeta) - c_\UEG \dashint_{\RR^3} \zeta^{4/3} \right| \le \epsilon \dashint_{\RR^3} \bigl(\zeta + \zeta^{4/3}\bigr) 
+ \frac{C_{p,\theta}}{\epsilon^b} \dashint_{\RR^3} |\grad\zeta^\theta|^p 
\end{equation}
holds true, where $b=\min\{2p-1,(1+3\theta)p-4\}$ and the positive constant $C_{p,\theta}$ only depends on $p$ and $\theta$. It may be bounded as $C_{p,\theta}\le 2.71 b \left(\frac{10}{\theta}\right)^b$.
\end{theorem}

We can conclude that for slowly varying inhomogeneities, the NUEG energy is well-approximated by the UEG energy via the LDA. Therefore, in this scenario a perturbative treatment (leading to the ``weakly non-uniform electron gas'') might be justified. This goal is not pursued in the present work.

Replacing $\zeta$ by $\zeta(\lambda\cdot)$ in the above theorem we obtain the convergence rate
$$
\left|e_\NUEG^\cl(\zeta(\lambda\cdot))-c_\UEG \dashint_{\RR^3} \zeta^{4/3}\right|\le C\lambda^{1/2} \left(\dashint_{\RR^3} \bigl(\zeta + \zeta^{4/3}\bigr)\right)^{1-1/2p}
\left(\dashint_{\RR^3} |\grad\zeta^\theta|^p \right)^{1/2p}.
$$
This form shows that indeed for a nonconstant inhomogeneity we get an energy different from the UEG energy, even in the slowly varying limit. 

The method of proof of \cref{cllda} is essentially the same as the one in \cite[Appendix A]{lewin2020local}. It is an important open problem to improve the exponent $\epsilon$ and the constant $C_{p,\theta}$.
We would like to point out that a more general result may be derived without the regularity assumption by following the ideas of \cite{jex2025classical}. 

\subsection{Quantum non-uniform electron gas}\label{qnuegsec}
Next, we introduce our definition of the quantum non-uniform gas and consider some of its basic properties. As mentioned above, we ignore spin for simplicity as (collinear) spin is straightforward to handle, but it would make proofs much more tedious. 

We call a function $\zeta:\RR^3\to\RR_+$ a \emph{(quantum) inhomogeneity} if $\sqrt{\zeta}\in H_\per^1(\LC)$. 
Recall that for such functions, the mean values $\dashint_{\RR^3} \zeta$ and $\dashint_{\RR^3} |\grad\sqrt{\zeta}|^2$ are well-defined and finite. 
By the Sobolev inequality, we have $\int_\Lambda \zeta^p<\infty$
for all $1\le p\le 3$. In particular, $\dashint_{\RR^3} \zeta^{5/3}$ and $\dashint_{\RR^3} \zeta^{4/3}$ are finite.

Next, we define the non-uniform electron gas energy per finite volume. We offer two definitions following \cite{lewin2018statistical,lewin2020local} which turn out to be equivalent.

\subsubsection{Definition based smeared indicators}
Recall the $N$-representability criterion from \cite{lieb1983density}, which says that a 
density $\rho\in L^1(\RR^3,\RR_+)$ with $\int_{\RR^3}\rho=N\in\NN$ comes from an $N$-particle wavefunction with finite kinetic energy precisely if $\sqrt{\rho}\in H^1(\RR^3)$. So in the quantum case, we need to cut off the inhomogeneity $\zeta$ smoothly near the boundary of our domain.

Let $0\le \eta\in C_c^\infty(\RR^3)$ be radial with $\supp\eta\subset B_1$ and $\int_{\RR^3}\eta=1$. Define
$\eta_\delta(x)=\delta^{-3} \eta(\delta^{-1}x)$.
Then $\supp\eta_\delta\subset B_{\delta}$ and $\int_{\RR^3} \eta_\delta=1$. We say that $\eta_\delta$ is a \emph{regularization function} (or mollifier) with smearing parameter $\delta>0$. While the finite-volume quantities are clearly dependent on the particular choice of the regularization function, this dependence will disappear when passing to the thermodynamic limit. 

Fix a bounded domain $\Omega\subset\RR^3$ and an inhomogeneity $\zeta$.
Clearly, $\sqrt{(\iden_{\Omega}*\eta_\delta) \zeta}\in H^1(\RR^3)$ so we may consider the function $f_{\Omega,\eta_\delta,\zeta}: \RR^3\to \RR$ given by
\begin{equation}\label{fdef}
f_{\Omega,\eta_\delta,\zeta}(a)= \frac{E^\hbar\bigl((\iden_{\Omega-a}*\eta_\delta)\zeta\bigr)}{|\Omega|}=\frac{E^\hbar\bigl((\iden_{\Omega}*\eta_\delta)\zeta(\cdot-a)\bigr)}{|\Omega|}.
\end{equation}
In the second equality, the translation-invariance of the Levy--Lieb functional $F_\LL^\hbar(\rho)$ and of the Coulomb self-energy $D(\rho)$ was used.

As $\zeta$ is $\LC$-periodic, it is obvious that $f_{\Omega,\eta_\delta,\zeta}$ is $\LC$-periodic. 
Moreover, we have the \emph{a priori} estimate $\dashint_{\RR^3} |f_{\Omega,\eta_\delta,\zeta}(a)|\,\dd a<\infty$ (see \cref{fabsapriori} below).
Using \cref{permeanthm}, we can deduce that the mean value of $f_{\Omega,\eta_\delta,\zeta}$ over $\RR^3$ exists and is finite.
Next, we average over all rotations $R\in\SO(3)$ in $\dashint_{\RR^3} f_{\Omega,\eta_\delta,\zeta(R\cdot)}$ to make our energy rotationally invariant as well. In summary, we arrive at
\begin{equation}\label{evoldef}
\boxed{e_{\Omega,\eta_\delta}^\hbar(\zeta)=\int_{\SO(3)}\dd  R \dashint_{\RR^3}\dd a\, \frac{E^\hbar\bigl((\iden_{\Omega}*\eta_\delta)\zeta(R(\cdot-a))\bigr)}{|\Omega|}}
\end{equation}
which by construction is isometry-invariant with respect to both $\Omega$ and $\zeta$, separately.
In \cref{aprioricoro} below, we state an \emph{a priori} bound on $e_{\Omega,\eta_\delta}^\hbar(\zeta)$.
Clearly, our definition collapses to the usual indirect energy per volume of the \emph{uniform} electron gas (cf. \cite{lewin2020local}) by setting $\zeta\equiv \rho_0$, i.e. 
$e_{\Omega,\eta_\delta}^\hbar(\rho_0)=|\Omega|^{-1}E^\hbar\bigl((\iden_{\Omega}*\eta_\delta)\rho_0\bigr)$.

\subsubsection{Definition based on a transition region}
Instead of using a smeared indicator function to achieve smooth cutoff, it is also reasonable to do the following. Let $\Omega\subset\RR^3$ be a bounded connected domain. For any $s>0$, introducing the inner set
$$
\Omega^{s-}=\{ x\in\Omega : \dist(x,\partial\Omega)\ge s \}
$$
and the outer set
$$
\Omega^{s+}=\Omega \cup \{x\in\RR^3 : \dist(x,\partial\Omega)\le s \},
$$
we can give another definition of the indirect energy per volume,
\begin{equation}\label{evoltrans}
\boxed{\wt{e}_{\Omega,s}^\hbar(\zeta)=\int_{\SO(3)}\dd  R \dashint_{\RR^3}\dd a\, \inf_{\substack{\sqrt{\rho}\in H^1(\RR^3)\\ \zeta(R(\cdot-a))\iden_{\Omega^{s-}}\le \rho \le \zeta(R(\cdot-a))\iden_{\Omega^{s+}} }} \frac{E^\hbar(\rho)}{|\Omega|}}
\end{equation}
Taking $\rho=(\iden_{\Omega}*\eta_\delta)\zeta(R(\cdot-a))$ for a sufficiently small $\delta>0$, we may use an \emph{a priori} upper bound (see \cref{basicbdprop} below), and for the lower bound using the Lieb--Oxford inequality, we obtain that the $a$-integrand is summable, hence again by \cref{permeanthm}
the above mean value exists and finite. 

Intuitively speaking, the above minimization allows the density to ``relax'' in the transition region $\Omega^{s+}\setminus\Omega^{s-}$. We will require that the scale $s$ of this transition region goes to infinity in the thermodynamic limit in such a way that the ratio $s/|\Omega|^{1/3}$ becomes negligible, for a regular domain sequence $\Omega$.\footnote{Recall that for regular domains $\diam(\Omega)\le C|\Omega|^{1/3}$, see \cite{fisher1964free}.} 
The $a$-integrand is again $\LC$-periodic because $\zeta$ is.
Clearly, $\wt{e}_{\Omega,s}^\hbar(\zeta)$ is isometry-invariant in both $\zeta$ and $\Omega$. 
The conceptual benefit of this definition over the previous one is that the ``shape'' of the
transition is allowed to be arbitrary. 

\subsubsection{Thermodynamic limit and equivalence of definitions}
Next, we consider the existence and equivalence of the thermodynamic limits for both definitions.

\begin{theorem}[Quantum non-uniform electron gas]\label{thermolimthm}
Fix an inhomogeneity $\sqrt{\zeta}\in H^1_\per(\LC)$.
Let $\{\Omega_N\}\subset\RR^3$ be a sequence of bounded connected domains with $|\Omega_N|\to\infty$, such that $\Omega_N$ has uniformly $\kappa$-regular boundary with $\kappa(t)=Ct$. 
\begin{itemize}
\item[(i)] Let $\{\delta_N\}\subset\RR_+$ be any sequence such that $\delta_N/|\Omega_N|^{1/3}\to 0$ and $\delta_N|\Omega_N|^{1/3}\to\infty$.
Then the following thermodynamic limit exists
\begin{equation}\label{nueglim}
\lim_{N\to\infty} e^\hbar_{\Omega_N,\eta_{\delta_N}}(\zeta)=e_\NUEG^\hbar(\zeta),
\end{equation}
and is independent of the sequences $\{\Omega_N\}$, $\{\delta_N\}$ and the regularization function $\eta$. 
\item[(ii)] Let $\{s_N\}\subset\RR_+$ be a sequence such that $s_N\to\infty$ and $s_N/|\Omega_N|^{1/3}\to 0$. Then the following thermodynamic limit exists
\begin{equation}\label{thlimtr}
\lim_{N\to\infty} \wt{e}_{\Omega_N,s_N}^\hbar(\zeta)=e_\NUEG^\hbar(\zeta),
\end{equation}
with the same limiting $e_\NUEG^\hbar(\zeta)$ as in (i). Again, the limit is independent of the sequences $\{\Omega_N\}$ and $\{s_N\}$.
\end{itemize}
\end{theorem}

Again, we make a few remarks.

\begin{remark}\label{qnuegrmrk}
\noindent (i) By dropping the Coulomb interaction, we obtain a similar result for the kinetic energy per volume $\tau^\hbar_{\Omega,\eta_{\delta}}(\zeta)$ and therefore 
the physically most interesting quantity, the \emph{exchange-correlation energy per volume} 
$$
e^{\hbar,\xc}_{\Omega,\eta_{\delta}}(\zeta)=e^\hbar_{\Omega,\eta_{\delta}}(\zeta)- \tau^\hbar_{\Omega,\eta_{\delta}}(\zeta)
$$
also admits a thermodynamic limit. 
\vspace{.3em}\\
\noindent (ii) The following \emph{a priori} bounds hold true (see \cref{aprioricoro}). For every $0< \epsilon\le 1/15$ and every $\sqrt{\zeta}\in H^1_\per(\LC)$ there holds
$$
e_\NUEG^\hbar(\zeta)\le (1+\epsilon) c_{\TF}\hbar^2  \dashint_{\RR^3} \zeta^{5/3}+ \frac{38\hbar^2}{15}\frac{1}{\epsilon} \dashint_{\RR^3} |\grad \sqrt{\zeta}|^2,
$$
where $c_\TF=\frac{3}{5}(2\pi)^2 (4\pi/3)^{-2/3}$ is the Thomas--Fermi constant. 
For every $0<\epsilon\le 3/5$ and every $\sqrt{\zeta}\in H^1_\per(\LC)$ we have
$$
e_\NUEG^\hbar(\zeta)\ge (1-\epsilon) c_{\TF}\hbar^2  \dashint_{\RR^3} \zeta^{5/3}  - c_{\LO}  \dashint_{\RR^3} \zeta^{4/3}-\frac{20\hbar^2}{27}\frac{1}{\epsilon} \dashint_{\RR^3} |\grad \sqrt{\zeta}|^2.
$$
\vspace{.3em}\\
\noindent (iii) In the uniform case, the above two bounds become
$$
c_{\TF}\hbar^2\rho_0^{5/3}-c_\LO\rho_0^{4/3}\le e_\UEG^\hbar(\rho_0)\le c_{\TF}\hbar^2\rho_0^{5/3}.
$$
\end{remark}

\subsubsection{Further results}
Next, we consider the quantum analog of \cref{cllda}. 

\begin{theorem}[Local density approximation]\label{qlda}
Let $p>3$ and $0<\theta<1$ such that $2\le p\theta\le 1 + p/2$. For every $\epsilon>0$ and inhomogeneity $\zeta$ such that for all $\grad\zeta^\theta\in\Lloc^p(\RR^3)$ the
bound
\begin{align*}
\left| e_\NUEG^\hbar(\zeta) - \dashint_{\RR^3} e_\UEG^\hbar(\zeta(x))\,\dd x\right|&\le \epsilon \dashint_{\RR^3} \bigl( \zeta + \zeta^2 \bigr)\\
&+\frac{C(1+\epsilon)}{\epsilon} \dashint_{\RR^3} |\grad\sqrt{\zeta}|^2 + \frac{C}{\epsilon^{4p-1}} \dashint_{\RR^3} |\grad\zeta^\theta|^p
\end{align*}
holds true for some universal constant $C>0$ which only depends on  $p$, $\theta$ and $\hbar$.
\end{theorem}

Plugging in $\zeta(\lambda\cdot)$ and $\epsilon=\lambda^\alpha$, we find the convergence rate 
\begin{multline*}
\left| e_\NUEG^\hbar(\zeta(\lambda\cdot)) - \dashint_{\RR^3} e_\UEG^\hbar(\zeta(x))\,\dd x\right|\le \lambda^{1/10} \dashint_{\RR^3} \bigl( \zeta + \zeta^2 \bigr)\\
+ C\lambda^{19/20}\left( \dashint_{\RR^3} |\grad\sqrt{\zeta}|^2 + \dashint_{\RR^3} |\grad\zeta^\theta|^p \right)
\end{multline*}
for the slowly varying limit $\lambda\to 0$. 

Next, we consider the question of obtaining $e_\NUEG^\cl(\zeta)$ as the semiclassical limit of $e_\NUEG^\hbar(\zeta)$ as $\hbar\to 0$. Note first that by scaling relation \cref{Ehbarscale}, we have for all $\lambda>0$
$$
e_{\Omega,\eta_\delta}^\hbar(\zeta)=\lambda^{4/3} e_{\lambda^{1/3}\Omega,\eta_{\lambda^{1/3}\delta}}^{\lambda^{1/6}\hbar}(\lambda^{-1}\zeta(\lambda^{-1/3}\cdot)).
$$
This suggests to write $\zeta=\rho_0\zeta_0(\rho_0^{1/3}\cdot)$, where $\dashint\zeta_0=\dashint \zeta_0(\rho_0^{1/3}\cdot)=1$ and $\rho_0=\dashint\zeta=\lambda$, so that in the thermodynamic limit
$$
e_{\NUEG}^\hbar(\zeta)=\rho_0^{4/3} e_\NUEG^{\rho_0^{1/6}\hbar}(\zeta_0)=\rho_0^{4/3} \wh{e}_\NUEG(\rho_0^{1/3}\hbar^2;\zeta_0).
$$
Here, we have set $\wh{e}_\NUEG(\mu;\zeta_0)=e_\NUEG^{\sqrt{\mu}}(\zeta_0)$ for all $\mu\ge 0$ and $\dashint\zeta_0=1$. 
For fixed $\zeta_0$, the function $\mu\mapsto\wh{e}_\NUEG(\mu;\zeta_0)$ is concave and increasing.

\begin{proposition}[Semiclassical bound]\label{semiclassthm}
For every normalized inhomogeneity $\sqrt{\zeta_0}\in H^1_\per(\LC)$ with $\dashint\zeta_0=1$, there holds
$$
\liminf_{\mu\to 0} \wh{e}_\NUEG(\mu;\zeta_0)\ge e_\NUEG^\cl(\zeta_0).
$$
\end{proposition}

It remains an open problem to obtain the opposite bound. The trial state construction
based on the regularization strategy of \cite{lewin2018statistical} does not work here as 
that would result in a mollified inhomogeneity. 

Finally, we mention that the functional $\zeta\mapsto e_\NUEG^\hbar(\zeta)$ inherits the weak*-lower semicontinuity property of the grand-canonical Levy--Lieb functional \cite{lewin2019universal}. See \cref{enueglsc} near the end of the paper for the precise statement and its proof.

\section{Proofs for the classical case}\label{clproofs}
The rest of the paper is devoted to the proofs. We first provide proofs for our results concerning the classical NUEG.
Throughout this section, we omit the designation ``$\cl$'' from $E^\cl(\rho)$ and $e_\Omega^\cl(\zeta)$ for clarity.

\subsection{Known bounds}
The following estimate on the indirect Riesz energy will be used repeatedly in the sequel. 
\begin{theorem}[Lieb--Oxford inequality]\label{lobound}
Let $0<s<d$. For any symmetric grand-canonical probability $\PP$, the bound
\begin{align*}
\Cc_s(\PP)-D_s(\rho_\PP)\ge - c_\LO(d,s) \int_{\RR^d} \rho_\PP^{1+s/d}
\end{align*}
holds true with some universal constant $c_\LO(d,s)>0$.
\end{theorem}
The proof is a straightforward modification of the one in \cite{lundholm2016fractional}. 
In the 3D Coulomb case, the constant $c_\LO=c_\LO(3,1)\le 1.58$ is the best known at present \cite{lewin2022improved}. The proof of \cite{lewin2022improved} works all the same in the grand-canonical setting.

For the next result, we recall the concept of localization of an $N$-particle symmetic probability $\PP$ on $\RR^{dN}$. See \cite{di2025grand} for more details. For a subset $A\subset\RR^d$, the \emph{localization} (or restriction) of $\PP$ onto $A$ is a grand-canonical symmetric probability $\PP|_A=((\PP|_A)_n)_{n\ge 0}$ and is given by 
\begin{align*}
(\PP|_A)_0&=\PP\bigl((\RR^d\setminus A)^N\bigr)\\
(\PP|_A)_n&\equiv \binom{N}{n} \PP\bigl(\cdot \times (\RR^d\setminus A)^{N-n}\bigr) \quad \text{on}\;A^n \quad (n=1,\ldots,N-1)\\
(\PP|_A)_N&=\PP \quad \text{on}\;A^N
\end{align*}
and $(\PP|_A)_n\equiv 0$ whenever $n\ge N+1$. 

Next, we specify a special tetrahedral tiling of $\RR^3$ that will be used throughout the proofs. Recall that $C_1=(-\frac{1}{2},\frac{1}{2})^3$ denotes the origin-centered open unit cube. Let $\tetra_j=T_j(\tetra)$ ($j=1,\ldots,24$) be a tiling of $C_1$ by 24 congruent open tetrahedra, where
the isometries $T_j$ are given by $T_j(x)=R_jx-z_j$ for appropriate $R_j\in \SO(3)$ and $z_j\in C_1$. Here, $\tetra$ is a ``reference'' tetrahedron with barycenter 0 and $T_j$ is the isometry that maps $\tetra$ to $\tetra_j$. 
Given a scale $\ell>0$, we may write 
$$
\RR^3=\bigcup_{z\in\ZZ^3}\ol{(C_\ell+\ell z)}=\bigcup_{z\in\ZZ^3} \bigcup_{j=1}^{24} \ol{( \ell \tetra_j + \ell z)}.
$$
The corresponding partition of unity is given by 
\begin{equation}\label{pouind}
\sum_{z\in\ZZ^3} \sum_{j=1}^{24} \iden_{\ell \tetra_j}(\cdot-\ell z)\equiv 1\;\text{(a.e.)}.
\end{equation}

We can now state the following very important result that allows us to spatially decouple the indirect 3D Coulomb energy on the tetrahedral tiling \cite{graf1995electrostatic,graf1995molecular,lewin2018statistical}.

\begin{theorem}[Graf--Schenker inequality]\label{gsineq}
For any $N$-particle symmetric probability density $\PP$ and any $\ell>0$ the bound 
\begin{multline*}
\Cc(\PP)-D(\rho_\PP)\ge \int_{\SO(3)}\dd R \dashint_{C_\ell}\dd\tau \sum_{z\in\ZZ^3}\sum_{j=1}^{24} 
\Big( \Cc(\PP|_{R^\top(\ell\tetra_j + \ell z + \tau)}) - D(\iden_{R^\top(\ell\tetra_j + \ell z + \tau)}\rho_{\PP}) \Big) \\
- \frac{c_\GS}{\ell} \int_{\RR^3} \rho_\PP
\end{multline*}
holds true with the constant 
$c_{\GS}=\pi\frac{1+2\sqrt{2}}{4}\approx 3.0068$. In particular, there exists an isometry $(R,\tau)\in \SO(3)\times C_\ell$ such that
\begin{equation}\label{gsmv}
\Cc(\PP)-D(\rho_\PP)\ge \sum_{z\in\ZZ^3}\sum_{j=1}^{24} 
\Big( \Cc(\PP|_{R^\top(\ell\tetra_j + \ell z + \tau)}) - D(\iden_{R^\top(\ell\tetra_j + \ell z + \tau)}\rho_{\PP}) \Big) 
- \frac{c_\GS}{\ell} \int_{\RR^3} \rho_\PP.
\end{equation}
\end{theorem}
The constant $c_\GS$ is obtained by a straightforward but lengthy calculation, using the explicit formulas of \cite{graf1995electrostatic,graf1995molecular}. 
A similar estimate holds true for grand-canonical probabilities as well (with the same constant). The second term on the r.h.s. is a kind of localization error analogous to the term arising from the localization of the kinetic energy (i.e. the IMS formula). 

We will use the following simple observations to replace a continuous inhomogeneity $\zeta$ locally in a bounded domain $\Omega\subset\RR^d$. Let 
$$
\ul{\zeta}=\min_{\ol{\Omega}}\zeta,\quad\text{and}\quad \ol{\zeta}=\max_{\ol{\Omega}}\zeta
$$
denote the minimum and maximum of $\zeta$ over $\ol{\Omega}$. Then by the subadditivity and nonpositivity of $E(\rho)$, we have
\begin{equation}\label{replmin}
E(\iden_\Omega \zeta)\le E(\iden_\Omega\ul{\zeta}),
\end{equation}
and
\begin{equation}\label{replmax}
E(\iden_\Omega\ol{\zeta})\le E(\iden_\Omega\zeta).
\end{equation}

In order to control the fluctuations of the density locally, we need the following form 
of Morrey's inequality.
\begin{theorem}[Morrey's inequality on tetrahedra]\label{morrineq}
For every $p>3$ and $u\in W^{1,p}(\tetra)$, where $\tetra\subset\RR^3$ is a tetrahedron isometric to a dilation of the tetrahedron $\tetra_1$ above, the bound
$$
|u(x)-u(y)|\le c_\Mo |x-y|^{1-3/p} \|\grad u\|_{L^p},
$$
holds true for all $x,y\in\tetra$, where the constant $c_\Mo>0$ only depends on $p$. The estimate $c_\Mo\le 2\cdot 24^{1/p}\frac{p-1}{p+1}$ holds.
\end{theorem}

The estimate on the constant is the one obtained in the proof of \cite[Lemma 4.28]{adams2003sobolev}.

\subsection{Thermodynamic limit for dyadic cubes}
To prove that the limit $e_{C_{2^N}}(\zeta)$ exists we first establish that the energy is nonincreasing for dyadic cubes.
\begin{lemma}\label{dyadmono}
$e_{C_{2^{N+1}}}(\zeta)\le e_{C_{2^N}}(\zeta)$ for every $N\ge 1$.
\end{lemma}
\begin{proof}
The big cube $C_{2^{N+1}}$ is essentially composed of $2^d$ smaller cubes of side length $2^N$ with centers placed at the vertices of $C_{2^N}$. More precisely, 
$$
\ol{C_{2^{N+1}}}=\bigcup_{z\in\{\pm 1\}^d} \ol{\bigl(C_{2^N}+2^{N-1}z\bigr)}.
$$
Hence, using the subadditivity and the isometry invariance of $E(\rho)$,
\begin{align*}
e_{C_{2^{N+1}}}(\zeta)&=\int_{\SO(d)}\dd R \dashint_{\RR^d}\dd a\, \frac{E\left(\sum_{z\in\{\pm 1\}^d} \iden_{C_{2^N}+2^{N-1}z} \zeta(R(\cdot-a)) \right)}{|C_{2^{N+1}}|}\\
&\le \frac{1}{2^d} \sum_{z\in\{\pm 1\}^d}  \int_{\SO(d)}\dd R \dashint_{\RR^d}\dd a\, \frac{E\left(\iden_{C_{2^N}+2^{N-1}z} \zeta(R(\cdot-a)) \right)}{|C_{2^{N}}|}\\
&=\frac{1}{2^d}\sum_{z\in\{\pm 1\}^d} \int_{\SO(d)}\dd R \dashint_{\RR^d}\dd a\,  \frac{E\left(\iden_{C_{2^N}} \zeta(R(\cdot+2^{N-1}z-a)) \right)}{|C_{2^{N}}|}\\
&=\frac{1}{2^d}\sum_{z\in\{\pm 1\}^d} \int_{\SO(d)}\dd R \dashint_{\RR^d}\dd a\, \frac{E\left(\iden_{C_{2^N}} \zeta(R(\cdot-a)) \right)}{|C_{2^{N}}|}=e_{C_{2^N}}(\zeta).
\end{align*}
In the next to last step we used the translational invariance of the mean value
to absorb $2^{N-1}z$ in the argument of $\zeta$.
\end{proof}

Hence, the numerical sequence $\{e_{C_{2^{N+1}}}(\zeta)\}_{N\ge 1}$ is nonincreasing and bounded from below according to \cref{ecllobound}. 
Therefore, the limit
$$
e_\NUEG^\cl(\zeta)=\lim_{N\to\infty} e_{C_{2^{N}}}(\zeta)
$$
exists and is finite.

\subsection{Thermodynamic limit for general domains}
First, we prove the upper bound via an inner approximation by dyadic cubes. Let 
$$
J_0^{n,N}=\{z\in\ZZ^d : (C_{2^{n}} + 2^n z) \subset \Omega_N\}
$$ 
and 
$$
\Omega_N'=\bigcup_{z\in J_0^{n,N}} \bigl(C_{2^{n}} + 2^n z\bigr)\subset \Omega_N.
$$
By the subadditivity and nonpositivity of $E(\rho)$, we have
\begin{align*}
\frac{E(\iden_{\Omega_N}\zeta)}{|\Omega_N|}&\le \frac{E(\iden_{\Omega_N'}\zeta)}{|\Omega_N|}+\frac{E(\iden_{\Omega_N\setminus\Omega_N'}\zeta)}{|\Omega_N|}\\
&\le\frac{1}{|\Omega_N|} E\Biggl( \sum_{z\in J_0^{n,N}} \iden_{C_{2^{n}} + 2^n z} \zeta \Biggr)\le \frac{1}{|\Omega_N|} \sum_{z\in J_0^{n,N}} E(\iden_{C_{2^{n}} + 2^n z} \zeta).
\end{align*}
After replacing $\zeta$ by $\zeta(R(\cdot-a))$ and averaging over all translations and rotations, we find
$$
e_{\Omega_N}(\zeta)\le \frac{2^n|J_0^{n,N}|}{|\Omega_N|} e_{C_{2^n}}(\zeta)=\frac{|\Omega_N'|}{|\Omega_N|}e_{C_{2^n}}(\zeta).
$$
To estimate the prefactor on the r.h.s, we use the uniform $\kappa$-regularity of the sequence $\Omega_N$ in the manner
\begin{align*}
|\Omega_N'|&=|\Omega_N| - \sum_{\substack{z\in\ZZ^3\\(C_{2^{n}} + 2^n z)\cap (\RR^d\setminus\Omega_N)\neq \emptyset}} |(C_{2^{n}} + 2^n z)\cap\Omega_N|\\
&\ge |\Omega_N| - |\{x\in\RR^d : d(x,\partial\Omega_N)\le 2^{n}\sqrt{d}\}|\\
&\ge |\Omega_N|\left(1 - \kappa(2^n\sqrt{d}|\Omega_N|^{-1/d})\right).
\end{align*}
Since $\kappa(t)\to 0$ as $t\to 0$, we obtain after letting $N\to\infty$ and then $n\to\infty$, that 
$$
\limsup_{N\to\infty} e_{\Omega_N}(\zeta)\le e_\NUEG^\cl(\zeta),
$$
where $e_\NUEG^\cl(\zeta)$ is the limit for dyadic cubes, as above.

It remains to prove the lower bound. By an appropriate translation of $\Omega_N$ we can assume that it is contained in a dyadic cube $C_{2^K}$ of side length $|\Omega_N|^{1/d}$. 
We now take 
$$
J_c^{n,N}=\{z\in\ZZ^d : (C_{2^{n}} + 2^n z) \subset C_{2^K}\setminus\Omega_N\}
$$ 
and 
$$
\Omega_N''=\bigcup_{z\in J_c^{n,N}} \bigl(C_{2^{n}} + 2^n z\bigr).
$$
Therefore, we may decompose $C_{2^K}=\Omega_N\cup \Omega_N'' \cup R_N$, where the remainder set is
$$
R_N=\bigcup_{\substack{z\in J_c^{n,N}\\ (C_{2^{n}} + 2^n z)\cap\partial\Omega_N\neq\emptyset}} \bigl(C_{2^{n}} + 2^n z\bigr)\cap \bigl(C_{2^K}\setminus\Omega_N\bigr).
$$
By \cref{dyadmono}, subadditivity and nonpositivity of $E(\rho)$, we have 
\begin{align*}
\frac{|C_{2^K}|}{|\Omega_N|} e_\NUEG^\cl(\zeta) \le \frac{|C_{2^K}|}{|\Omega_N|}  e_{C_{2^K}}(\zeta)\le  e_{\Omega_N}(\zeta) 
+  \frac{|\Omega_N''|}{|\Omega_N|} e_{C_{2^n}}(\zeta)
\end{align*}
Rearranging, we obtain
$$
e_\NUEG^\cl(\zeta) + \frac{|\Omega_N''|}{|\Omega_N|}(e_\NUEG^\cl(\zeta)-e_{C_{2^n}}(\zeta)) + \frac{|R_N|}{|\Omega_N|} e_\NUEG^\cl(\zeta) 
\le e_{\Omega_N}(\zeta).
$$
Here, using the uniform $\kappa$-regularity of $\{\Omega_N\}$, we have 
\begin{align*}
|R_N|&=\sum_{\substack{z\in J_c^{n,N}\\ (C_{2^{n}} + 2^n z)\cap\partial\Omega_N\neq\emptyset}} \bigl|\bigl(C_{2^{n}} + 2^n z\bigr)\cap \bigl(C_{2^K}\setminus\Omega_N\bigr)\bigr|\\
&\le |\{x\in\RR^d : d(x,\partial\Omega_N)\le 2^{n}\sqrt{d}\}|\\
&\le |\Omega_N| \kappa\bigl(2^n\sqrt{d}|\Omega_N|^{-1/d}\bigr).
\end{align*}
Clearly, $|\Omega_N''|\le |C_{2^K}|\le C|\Omega_N|$ and so in total we get the estimate
$$
e_\NUEG^\cl(\zeta) + C\bigl(e_\NUEG^\cl(\zeta)-e_{C_{2^n}}(\zeta)\bigr) + \kappa\bigl(2^n\sqrt{d}|\Omega_N|^{-1/d}\bigr) e_\NUEG^\cl(\zeta)\le e_{\Omega_N}(\zeta).
$$
Taking the liminf as $N\to\infty$ and $n\to\infty$, we arrive at 
$$
e_\NUEG^\cl(\zeta)\le \liminf_{N\to\infty} e_{\Omega_N}(\zeta).
$$

In summary, the limits $\lim_{N\to\infty} e_{\Omega_N}(\zeta)$ and $\lim_{N\to\infty} e_{C_{2^N}}(\zeta)$ exist and equal to a common number
$e_\NUEG^\cl(\zeta)$. This completes the proof of \cref{clnuegthm}.

\subsection{Convergence rate for tetrahedra}
In the 3D Coulomb case, when the domain is a dilated (reference) tetrahedron, we can 
obtain information about the rate of convergence $e_{\ell\tetra}(\zeta)\to e_\NUEG(\zeta)$ as $\ell\to\infty$. 

\begin{theorem}\label{clconvrate}
Suppose that $d=3$ and $s=1$. For $\ell>0$ sufficiently large, the bounds
$$
e_\NUEG(\zeta)\le e_{\ell\tetra}(\zeta)\le e_\NUEG(\zeta) + \frac{c_\GS}{\ell} \dashint_{\RR^3}\zeta
$$
hold true.
\end{theorem}

\begin{proof}
Using the subadditivity and nonpositivity of $E(\rho)$, 
\begin{align*}
e_{\ell'\tetra}(\zeta)&=\dashint_{\RR^3}\dd a\int_{\SO(3)}\dd R\, \frac{E\bigl(\iden_{\ell'\tetra}\zeta(R(\cdot-a))\bigr)}{|\ell'\tetra|}\\
&\le \frac{1}{|\ell'\tetra|} \sum_{(z,j)\in J_0} \dashint_{\RR^3}\dd a\int_{\SO(3)}\dd R\, E\bigl(\iden_{\ell\tetra_j+\ell z}\zeta(R(\cdot-a))\bigr)\\
&=\frac{|J_0|}{|\ell'\tetra|} \dashint_{\RR^3}\dd a\int_{\SO(3)}\dd R\, E\bigl(\iden_{\ell\tetra}\zeta(R(\cdot-a))\bigr)\\
&=\frac{|J_0||\ell\tetra|}{|\ell'\tetra|} e_{\ell\tetra}(\zeta)
\end{align*}
where $J_0=\{(z,j)\in\ZZ^3\times\{1,\ldots,24\} : (\ell\tetra_j+\ell z)\subset \ell'\tetra\}$.
Recall that $e_{\ell\tetra}(\zeta)\le 0$, so using the estimate 
$$
\frac{|J_0||\ell\tetra|}{|\ell'\tetra|}\ge \frac{\ell'^3 - C\ell'^2\ell}{\ell'^3}=1-\frac{C\ell}{\ell'},
$$
and taking the limit $\ell'\to\infty$ gives the stated lower bound.

For the upper bound using the Graf--Schenker inequality (\cref{gsineq}),
\begin{align*}
&e_{\ell'\tetra}(\zeta)=\int_{\SO(3)}\dd Q\dashint_{\RR^3}\dd a\,\frac{E\bigl(\iden_{\ell'\tetra}\zeta(Q(\cdot-a))\bigr)}{|\ell'\tetra|} \\
&\ge \frac{1}{|\ell'\tetra|} \int_{\SO(3)}\dd Q\dashint_{\RR^3}\dd a \int_{\SO(3)}\dd R \dashint_{C_\ell}\dd\tau \sum_{z\in\ZZ^3}\sum_{j=1}^{24} 
E\bigl(\iden_{\ell'\tetra}\iden_{\ell\tetra_j}(R\cdot-\ell z-\tau)\zeta(Q(\cdot-a))\bigr) \\
&\quad -\frac{c_\GS}{\ell} \frac{1}{|\ell'\tetra|} \int_{\SO(3)}\dd Q\dashint_{\RR^3}\dd a \int_{\RR^3} \iden_{\ell'\tetra}\zeta(Q(\cdot-a))=:\mathrm{(I)} + \mathrm{(II)}.
\end{align*}

Here, (II) is simply $-\frac{C}{\ell} \dashint_{\RR^3}\zeta$. To deal with (I), we define the index set
$$
J:=\left\{ \begin{aligned}(z,j)\in \ZZ^3\times\{1,\ldots,24\} : &\supp( \iden_{\ell\tetra_j}(R\cdot - \ell z - \tau)) \cap \supp(\iden_{\ell'\tetra})\neq \emptyset,\\
&\text{for some}\;R\in\SO(3)\;\text{and}\; \tau\in C_\ell \end{aligned} \right\},
$$
which collects all the small tetrahedra $R^\top(\ell\tetra_j + \ell z + \tau)$ that possibly intersect with the big tetrahedron $\ell'\tetra$. 
Furthermore, we define
$$
J_0:=\left\{ \begin{aligned}(z,j)\in \ZZ^3\times\{1,\ldots,24\} : &\supp( \iden_{\ell\tetra_j}(R\cdot - \ell z - \tau)) \subset \ell'\tetra,\\
&\text{for all}\;R\in\SO(3)\;\text{and}\; \tau\in C_\ell \end{aligned} \right\},
$$
which contains the small tetrahedra that are inside $\ell'\tetra$.

The $J_0$-part of (I) reads
\begin{align*}
\mathrm{(Ia)}&=\frac{1}{|\ell'\tetra|} \int_{\SO(3)}\dd Q\dashint_{\RR^3}\dd a \int_{\SO(3)}\dd R \dashint_{C_\ell}\dd\tau \sum_{(z,j)\in J_0}
E\bigl(\iden_{\ell\tetra_j}(R\cdot-\ell z-\tau)\zeta(Q(\cdot-a))\bigr)\\
&=\frac{1}{|\ell'\tetra|} \int_{\SO(3)}\dd Q\dashint_{\RR^3}\dd a \sum_{(z,j)\in J_0}
E\bigl(\iden_{\ell\tetra}\zeta(Q(\cdot-a))\bigr)\\
&=\frac{|J_0||\ell\tetra|}{|\ell'\tetra|} e_{\ell\tetra}(\zeta).
\end{align*}
The $J\setminus J_0$-part of (I) can be bounded by the Lieb--Oxford inequality as
\begin{align*}
\mathrm{(Ib)}&=\frac{1}{|\ell'\tetra|} \int_{\SO(3)}\dd Q\dashint_{\RR^3}\dd a \int_{\SO(3)}\dd R \dashint_{C_\ell}\dd\tau\\
&\quad\quad \times \sum_{(z,j)\in J\setminus J_0}
E\bigl(\iden_{\ell'\tetra}\iden_{\ell\tetra_j}(R\cdot-\ell z-\tau)\zeta(Q(\cdot-a))\bigr)\\
&\ge - c_\LO \frac{1}{|\ell'\tetra|} \int_{\SO(3)}\dd Q\dashint_{\RR^3}\dd a \int_{\SO(3)}\dd R \dashint_{C_\ell}\dd\tau \\
&\quad\quad\times\sum_{(z,j)\in J\setminus J_0}\int_{\RR^3} \iden_{\ell'\tetra}\iden_{\ell\tetra_j}(R\cdot-\ell z-\tau)\zeta(Q(\cdot-a))^{4/3}\\
&=- c_\LO \frac{1}{|\ell'\tetra|} \int_{\SO(3)}\dd R \dashint_{C_\ell}\dd\tau \sum_{(z,j)\in J\setminus J_0}\int_{\RR^3} \iden_{\ell'\tetra}\iden_{\ell\tetra_j}(R\cdot-\ell z-\tau) \dashint_{\RR^3} \zeta^{4/3}\\
&\ge - c_\LO \frac{|\ell\tetra||J\setminus J_0|}{|\ell'\tetra|}\dashint_{\RR^3} \zeta^{4/3}.
\end{align*}
Clearly, the volume ratios may be bounded as
$$
\frac{|J_0||\ell\tetra|}{|\ell'\tetra|}\le 1 \quad\text{and}\quad \frac{|\ell\tetra||J\setminus J_0|}{|\ell'\tetra|}\le\frac{C\ell\ell'^2}{\ell'^3}=\frac{C\ell}{\ell'}.
$$
Hence, in the limit $\ell'\to\infty$ we obtain the estimate
$$
e_\NUEG(\zeta)\ge e_{\ell\tetra}(\zeta) - \frac{c_\GS}{\ell} \dashint_{\RR^3}\zeta,
$$
which finishes the proof. 
\end{proof}

\subsection{Local density approximation}\label{clldaproof}
As mentioned above, the proof of \cref{cllda} is very similar to that of \cite[Theorem 4]{lewin2020local}. Instead of explaining the differences, we adapt the proof technique to our case and
produce explicit constants.

If $\epsilon$ is sufficiently large, the estimate follows from the Lieb--Oxford inequality. More precisely,
from \cref{ecllobound},
\begin{align*}
0\ge e_{\ell\tetra}^\cl(\zeta)&\ge -c_\LO \dashint_{\RR^3} \zeta^{4/3}= c_\UEG \dashint_{\RR^3} \zeta^{4/3} - (c_\LO-|c_\UEG|) \dashint_{\RR^3} \zeta^{4/3},
\end{align*}
hence the stated estimate \cref{clldaest} follows from this, whenever $\epsilon\ge c_\LO-|c_\UEG|$.
Note that $c_\LO\ge |c_\UEG|$ according to \cite{lewin2018statistical}.

So we may assume that $\epsilon<\epsilon_0:=c_\LO-|c_\UEG|$. First, we prove the upper bound. By subadditivity and nonpositivity of $E(\rho)$ and \cref{replmin}, 
\begin{equation}\label{clldapre}
E(\iden_{\ell'\tetra}\zeta)\le \sum_{(z,j)\in J_0} E(\iden_{\ell\tetra_j+\ell z}\zeta)
\le \sum_{(z,j)\in J_0} E\bigl(\iden_{\ell\tetra_j+\ell z}\ul{\zeta}_{z,j}\bigr),
\end{equation}
where $J_0=\{(z,j)\in\ZZ^3\times\{1,\ldots,24\} : (\ell\tetra_j+\ell z)\subset \ell'\tetra\}$
and $\ul{\zeta}_{z,j}=\min_{\ol{\ell\tetra_j+\ell z}}\zeta$.
A special case of \cref{clconvrate} reads
\begin{equation}\label{cluegapri}
c_\UEG \rho_0^{4/3} \le \frac{E(\iden_{\ell\tetra}\rho_0)}{|\ell\tetra|} \le c_\UEG\rho_0^{4/3} + \frac{c_\GS}{\ell} \rho_0.
\end{equation}
Inserting the upper bound in our estimate \cref{clldapre},
\begin{align*}
E(\iden_{\ell'\tetra}\zeta)&\le c_\UEG \sum_{(z,j)\in J_0} \int_{\ell\tetra_j+\ell z} \ul{\zeta}_{z,j}^{4/3}  
+ \frac{c_\GS}{\ell} \sum_{(z,j)\in J_0} \int_{\ell\tetra_j+\ell z} \ul{\zeta}_{z,j}\\
&= c_\UEG \int_{\ell'\tetra} \zeta^{4/3} + 
|c_\UEG| \sum_{(z,j)\in J_0} \int_{\ell\tetra_j+\ell z} \left(\zeta^{4/3}-\ul{\zeta}_{z,j}^{4/3}\right) + \frac{c_\GS}{\ell}\int_{\ell'\tetra} \zeta.
\end{align*}
In the second term, the integral measuring the local fluctuations of the inhomogeneity is estimated using Morrey's and H\"older's inequality,
\begin{multline}\label{locflucest}
\int_{\ell\tetra_j+\ell z} \left(\zeta^{4/3}-\ul{\zeta}_{z,j}^{4/3}\right)\le\\
\frac{a}{p}\left(\frac{4}{3\theta}\right)^{\frac{p}{a}} c_\Mo^p |\tetra| \frac{1}{\epsilon^{p+\frac{p}{a}-1}} \int_{\ell\tetra_j+\ell z} |\grad\zeta^\theta|^p
+ \left(1-\frac{a}{p}\right)\epsilon \int_{\ell\tetra_j+\ell z} \zeta^{\left(\frac{4}{3}-\theta a\right)\frac{p}{p-a}}
\end{multline}
for all $0<a\le 1$ and we have used $\epsilon=1/\ell$. To proceed, we require that
$1\le \left(\frac{4}{3}-\theta a\right)\frac{p}{p-a}\le \frac{4}{3}$, or equivalently that
$$
\frac{4}{3}\le \theta p\le 1 + \frac{p}{3a},
$$
so that we can bound pointwise as $\zeta^{\left(\frac{4}{3}-\theta a\right)\frac{p}{p-a}}\le \zeta + \zeta^{4/3}$. 
Moreover, we replace $\epsilon\to \left(|c_\UEG|\left(1-\frac{a}{p}\right)+c_\GS\right)^{-1} \epsilon$ to absorb
the prefactor in the $\epsilon$ term.

After replacing $\zeta$ by $\zeta(R(\cdot-a))$ and averaging over all rotations and translations, and taking the limit $\ell'\to\infty$, we arrive at
$$
e_{\NUEG}(\zeta)\le c_\UEG \dashint_{\RR^3} \zeta^{4/3} + \epsilon \dashint_{\RR^3} \left(\zeta+\zeta^{4/3}\right)
+  \frac{C_{p,\theta,a}}{\epsilon^{p+\frac{p}{a}-1}}  \dashint_{\RR^3} |\grad\zeta^\theta|^p,
$$
where
$$
C_{p,\theta,a}=\frac{a}{p}\left(\frac{4}{3\theta}\right)^{\frac{p}{a}} \frac{2^p}{e^2} \left( \frac{3}{5} \left(\frac{9\pi}{3}\right)^{1/3}+\pi\frac{1+2\sqrt{2}}{4}\right)^{p+\frac{p}{a}-1}
$$
and we used the Lieb--Narnhofer bound $|c_\UEG| \le \frac{3}{5} \left(\frac{9\pi}{3}\right)^{1/3}$ from \cite{lieb1975thermodynamic} and $\left(\frac{p-1}{p+1}\right)^p\le e^{-2}$. 
We replace $C_{p,\theta,a}$ with its rough bound
\begin{equation}\label{Crough}
C_{p,\theta,a}\le \frac{15}{8} \left(\frac{10}{\theta}\right)^{p+\frac{p}{a}-1}.
\end{equation}
The optimal choice for $a$ minimizing the exponent of $\epsilon$ is
$$
a=\begin{cases}
\frac{p}{3(\theta p -1)} & \theta p \le 1 + \frac{p}{3}\\ 
1 & \text{otherwise}
\end{cases}
$$
which gives the stated exponent $b$.

The proof of the lower bound is less direct. The form \cref{gsmv} of the Graf--Schenker inequality implies that
\begin{equation}\label{clldagsch}
E(\iden_{\ell'\tetra}\zeta)\ge \sum_{z\in\ZZ^3}\sum_{j=1}^{24} E\bigl(\iden_{\ell\tetra_j}(\cdot-\ell z)\iden_{\ell'\tetra}\zeta\bigr) - \frac{c_\GS}{\ell}\int_{\RR^3}\iden_{\ell'\tetra}\zeta,
\end{equation}
where we absorbed the isometry $(R,\tau)$ into $\zeta$ (recall that $e_{\ell'\tetra}(\zeta)$ is isometry-invariant in $\zeta$). 
We would like to use \cref{replmax} to replace $\zeta$ by its maximum 
$\ol{\zeta}_{z,j}=\max_{\ol{\ell\tetra_j+\ell z}}\zeta$,
and after applying the lower bound of \cref{cluegapri}, proceed in a similar way as above. 
However, this does not work because we cannot locally control $\ol{\zeta}_{z,j}$ by (possibly a constant times) $\zeta(x)$.
To continue, let
$$
J=\{ (z,j)\in\ZZ^3\times\{1,\ldots,24\} : (\ell\tetra_j+\ell z)\cap(\ell'\tetra)\neq\emptyset \}
$$
and $J_0$ exactly as above. 
\vspace{.3em}\\
\noindent\textbf{Case I. (Simple tetrahedra)} Define the set $J_0^{\mathrm{S}}$ of all tetrahedra $(z,j)\in J_0$ such that 
\begin{equation}\label{clldapr}
c_\LO \int_{\ell\tetra_j+\ell z}\zeta^{4/3}\le \epsilon \int_{\ell\tetra_j+\ell z}\left(\zeta +\zeta^{4/3}\right) + \frac{A}{\epsilon^{p+p/a-1}} \int_{\ell\tetra_j+\ell z}|\grad\zeta^\theta|^p,
\end{equation}
where $A>0$ is a universal constant depending only on $p$, $\theta$ and $a$, which we specify later. Recall that we put $\epsilon=1/\ell$. 
The terms in \cref{clldagsch} belonging to $J_0^{\mathrm{S}}$ can be bounded as
\begin{multline*}
\sum_{(z,j)\in J_0^{\mathrm{S}}} E\bigl(\iden_{\ell\tetra_j}(\cdot-\ell z)\zeta\bigr)\ge - c_\LO \int_{\ell\tetra_j+\ell z} \zeta^{4/3}\\
\ge \epsilon \sum_{(z,j)\in J_0^{\mathrm{S}}} \int_{\ell\tetra_j+\ell z}\left(\zeta +\zeta^{4/3}\right) + \frac{A}{\epsilon^{p+p/a-1}} \sum_{(z,j)\in J_0^{\mathrm{S}}} \int_{\ell\tetra_j+\ell z}|\grad\zeta^\theta|^p
\end{multline*}
by the Lieb--Oxford inequality and \cref{clldapr}.
Clearly, if $\ol{\zeta}_{z,j}\le (\epsilon/c_\LO)^3$, then $(z,j)\in J_0^{\mathrm{S}}$. 
Tetrahedra $\ell\tetra_j+\ell z$ for which \cref{clldapr} holds are called \emph{simple}, following the terminology of \cite{jex2025classical}.
\vspace{.3em}\\
\noindent\textbf{Case II. (Main tetrahedra)} The alternative is that
$$
c_\LO \int_{\ell\tetra_j+\ell z}\zeta^{4/3}> \epsilon \int_{\ell\tetra_j+\ell z}\left(\zeta +\zeta^{4/3}\right) + \frac{A}{\epsilon^{p+p/a-1}} \int_{\ell\tetra_j+\ell z}|\grad\zeta^\theta|^p,
$$
and 
$$
\ol{\zeta}_{z,j}>(\epsilon/c_\LO)^3
$$
holds for $\ell\tetra_j+\ell z$, which we call a \emph{main tetrahedron} (again following \cite{jex2025classical}). The collection of main tetrahedra is denoted by $J_0^{\mathrm{M}}$.
The key idea of the proof technique in \cite{lewin2020local,jex2025classical} is that in a main tetrahedron the inhomogeneity 
is slowly varying. In fact, we have using Morrey's inequality (\cref{morrineq}) that
\begin{align*}
\ol{\zeta}_{z,j}^\theta-\ul{\zeta}_{z,j}^\theta&\le c_\Mo
\ell^{1-p/3}\left(\int_{\ell\tetra_j+\ell z}|\grad\zeta^\theta|^p\right)^{1/p}\\
&\le A^{-1/p} c_\Mo \epsilon^{\frac{1}{a}-\frac{1}{p}}\ell^{-\frac{p}{3}} \left(c_\LO\int_{\ell\tetra_j+\ell z}\zeta^{4/3}\right)^{1/p}\\
&\le A^{-1/p} c_\Mo c_\LO^{1/p} \epsilon^{\frac{1}{a}-\frac{1}{p}} \ol{\zeta}_{z,j}^{\frac{4}{3p}}
\end{align*}
Our hypothesis implies that $1<(c_\LO^3\epsilon^{-3})^{\theta-\frac{4}{3p}} \ol{\zeta}_{z,j}^{\theta-\frac{4}{3p}}$. Inserting this into the r.h.s. above gives
\begin{align*}
\ol{\zeta}_{z,j}^\theta-\ul{\zeta}_{z,j}^\theta\le A^{-1/p} c_\Mo c_\LO^{3\theta-\frac{3}{p}}  \epsilon^{\frac{1}{a}+\frac{3}{p}-3\theta} \ol{\zeta}_{z,j}^{\theta}
\end{align*}
Here, the exponent of $\epsilon$ is positive by our assumptions. 
Hence, we need
$$
A^{-1/p} c_\Mo c_\LO^{3\theta-\frac{3}{p}} \epsilon^{\frac{1}{a}+\frac{3}{p}-3\theta}<1,
$$
where the exponent of $\epsilon$ is nonnegative. Recall that $\epsilon<\epsilon_0$, we can choose for instance
$$
A=2 c_\Mo^p c_\LO^{3\theta p-3} \epsilon_0^{p(\frac{1}{a}+\frac{3}{p}-3\theta)},
$$
then 
$$
\ol{\zeta}_{z,j}\le \frac{1}{1-A^{-1/p} c_\Mo c_\LO^{3\theta-\frac{3}{p}} \epsilon^{\frac{1}{a}+\frac{3}{p}-3\theta}} \ul{\zeta}_{z,j}\le \frac{1}{1-2^{-1/p}}\ul{\zeta}_{z,j}.
$$
To summarize, 
$$
\ol{\zeta}_{z,j}\le (1-2^{-1/p})^{-1}\zeta(x)
$$
for $x\in \ol{\ell\tetra_j + \ell z}$ for any $(z,j)\in J_0^{\mathrm{M}}$.
The main tetrahedra contribute to the sum in \cref{clldagsch} according to
\begin{multline*}
\sum_{(z,j)\in J_0^{\mathrm{M}}} E\bigl(\iden_{\ell\tetra_j}(\cdot-\ell z)\iden_{\ell'\tetra}\ol{\zeta}_{z,j}\bigr) \ge c_\UEG \sum_{(z,j)\in J_0^{\mathrm{M}}} \int_{\ell\tetra_j+\ell z}\ol{\zeta}_{z,j}^{4/3} \\
\ge c_\UEG \sum_{(z,j)\in J_0^{\mathrm{M}}} \int_{\ell\tetra_j+\ell z}\zeta^{4/3} -
|c_\UEG| \sum_{(z,j)\in J_0^{\mathrm{M}}} \int_{\ell\tetra_j+\ell z} \left(\ol{\zeta}_{z,j}^{4/3}-\zeta^{4/3}\right),
\end{multline*}
where we used \cref{replmax} and the lower bound of \cref{cluegapri}.
On the fluctuation term, we apply an estimate like \cref{locflucest} but with the constant $C_{p,\theta,a}$ multiplied by $(1-2^{-1/p})^{-1}>1$. Using the rough estimate $A\le 1.90 \cdot 4^b$, we find that the estimate \cref{Crough} is larger, hence our final constant in front of the gradient term reads
$2.71 b \left(\frac{10}{\theta}\right)^b$.

Dividing \cref{clldagsch} by the volume and averaging $\zeta$ over all rotations and translations, the boundary terms are bounded in the usual manner
\begin{align*}
\frac{1}{|\ell'\tetra|}\dashint_{\RR^3}&\dd a\int_{\SO(3)}\dd R\sum_{(z,j)\in J\setminus J_0} E\bigl(\iden_{\ell\tetra_j}(\cdot-\ell z)\iden_{\ell'\tetra}\zeta(R(\cdot-a))\bigr)\\
&\ge - \frac{c_\LO}{|\ell'\tetra|} \sum_{(z,j)\in J\setminus J_0} \dashint_{\RR^3}\dd a\int_{\SO(3)}\dd R\int_{\RR^3} \iden_{\ell\tetra_j}(\cdot-\ell z)\zeta(R(\cdot-a))^{4/3}\\
&= - c_\LO \frac{|J\setminus J_0||\ell\tetra|}{|\ell'\tetra|} \dashint_{\RR^3} \zeta^{4/3}\\
&\ge -\frac{C\ell}{\ell'}\dashint_{\RR^3} \zeta^{4/3},
\end{align*}
which disappears in the limit $\ell'\to\infty$.
We arrive at 
\begin{align*}
e_{\NUEG}^\cl(\zeta)\ge c_\UEG\dashint_{\RR^3} \zeta^{4/3} - \epsilon \dashint_{\RR^3} \left(\zeta+\zeta^{4/3}\right) - \frac{2.71 b \left(\frac{10}{\theta}\right)^b}{\epsilon^{b}} \dashint_{\RR^3} |\grad\zeta^\theta|^p.
\end{align*}
by taking the limit $\ell'\to\infty$.
This completes the proof of the theorem.

\section{Proofs for the quantum case}\label{qproofs}
For simplicity of notation, we drop the superscript $\hbar$ in $E^\hbar$ and $e^\hbar$ when no confusion can arise. 

\subsection{Known bounds}\label{qknownbd}
We begin with a Lieb--Thirring kinetic energy inequality with semiclassical leading term and a gradient correction term. The original inequality is due to Nam \cite{nam2018lieb} and 
here we state a special case of an improved version by Seiringer and Solovej \cite{seiringer2023simple}.
\begin{theorem}[Semiclassical Lieb--Thirring inequality with gradient correction]\label{ltsemi}
For any $0<\epsilon\le 3/5$ and any one-particle fermionic density matrix  $0\le\gamma\le\iden$ on $L^2(\RR^3)$ the bound
$$
\Tr(-\lapl\gamma)\ge (1-\epsilon) c_\TF \int_{\RR^3} \rho_\gamma^{5/3} - \frac{10}{27} \frac{1}{\epsilon}\int_{\RR^3} |\grad\sqrt{\rho_\gamma}|^2 
$$
holds true.
\end{theorem}

Next, we recall an estimate that gives an \emph{a priori} upper bound on the kinetic energy functional $T(\rho)$. It was first conjectured by March and Young \cite{march1958variational} in 1958. Bokanowski, Grebert and Mauser \cite{bokanowski2003local} proved a version in the periodic setting. 
The proof was finally achieved by Lewin, Lieb and Seiringer \cite{lewin2020local} by a technique interesting in its own right. It was recently generalized by the present authors and E.~I.~Tellgren to the magnetic case, where the so-called paramagnetic current density is also prescribed \cite{csirik2024thermodynamic}.

\begin{theorem}[Lewin--Lieb--Seiringer]\label{llsbound}
Suppose that $\rho\in L^1(\RR^3;\RR_+)$ and $\sqrt{\rho}\in H^1(\RR^3)$. Then there exists a one-particle fermionic density matrix $0\le\gamma\le\iden$ on $L^2(\RR^3)$ such that $\rho_\gamma=\rho$ and for every $0<\epsilon\le 1/15$ the bound
$$
\Tr(-\lapl\gamma)\le (1+\epsilon)c_\TF \int_{\RR^3} \rho^{5/3} + \frac{19}{15}\frac{1}{\epsilon} \int_{\RR^3} |\grad\sqrt{\rho}|^2
$$
holds true.
\end{theorem}

The result immediately implies an \emph{a priori} bound on the kinetic energy functional
$T(\rho)$.
By representing the one-particle density matrix $\gamma$ by a quasi-free state $\vGamma$ (see \cite{bach1994generalized}), and using the fact that quasi-free states have nonpositive indirect Coulomb energy (just like Slater determinants), the following very useful \emph{a priori} estimate on the quantum indirect energy may be obtained.
\begin{corollary}\label{llsebound}
Suppose that $\rho\in L^1(\RR^3;\RR_+)$ and $\sqrt{\rho}\in H^1(\RR^3)$. Then for every $0<\epsilon\le 1/15$ the bound
$$
E(\rho)\le (1+\epsilon)c_\TF\hbar^2 \int_{\RR^3} \rho^{5/3} + \frac{19}{15}\frac{\hbar^2}{\epsilon} \int_{\RR^3} |\grad\sqrt{\rho}|^2
$$
holds true. 
\end{corollary}

We proceed by recalling a spatial decoupling lower bound on the indirect quantum energy.  In order to state this result, we need to use smeared indicators as sharp localization would blow up the kinetic energy.

Let the function $\eta\in C_c^\infty(\RR^3,\RR_+)$ be radial such that $\supp\eta\subset B_1$ and $\int_{\RR^3}\eta=1$. Define
$\eta_\delta(x)=\frac{1}{(b\delta)^3}\eta\left(\frac{x}{b\delta}\right)$,
for some $0<b<1$.
Then $\supp\eta_\delta\subset B_{b\delta}$ and $\int_{\RR^3} \eta_\delta=1$.

Let $\Omega\subset\RR^3$ be an open set. 
It is straightforward to see that
\begin{equation}\label{smearcases}
(\iden_\Omega * \eta_\delta)(x)=\int_\Omega \eta_\delta(\cdot-x)=\begin{cases}
1 & \supp\eta_\delta(\cdot-x)\subset \Omega \, ,\\
\in[0,1] & \supp\eta_\delta(\cdot-x)\cap \partial\Omega\neq \emptyset\, ,\\
0 & \supp\eta_\delta(\cdot-x)\subset \RR^3\setminus \Omega \,.
\end{cases}
\end{equation}
In particular, since $\supp\eta_\delta\subset B_{b\delta}$
\begin{equation}\label{smeardindunity}
\iden_{\Omega} * \eta_\delta\equiv 1 \;\text{in}\; \Omega\setminus (\partial\Omega+B_{b\delta} ).
\end{equation}
More concretely, let $\Omega=\ell S$, where $\ell>0$ and $S\subset\RR^3$ is a bounded star-shaped set with center 0.
Then \cref{smeardindunity} becomes the convenient\footnote{This is why we need the factor $0<b<1$ in the definition of $\eta_\delta$.} fact that
$$
\iden_{\ell S} * \eta_\delta\equiv 1 \;\text{in}\; (\ell-\delta)S.
$$
We refer to this fact as $\iden_{\ell S} * \eta_\delta\equiv 1$ \emph{well inside $\ell S$}.

Clearly, the integral of the smeared indicator of $\Omega$ is
$\int_{\RR^3} \iden_\Omega * \eta_\delta = |\Omega|$.
Next, we state an elementary, but very important kinetic energy estimate.

\begin{lemma}[Bound on kinetic energy of a smeared set]\label{indgradlemma}
Suppose that $\Omega\subset\RR^3$ is an open set with a boundary having finite 2-dimensional Lebesgue measure: $|\partial\Omega|<\infty$. 
Then the following bound holds true
$$
\int_{\RR^3} \left|\grad \sqrt{\iden_\Omega*\eta_\delta}\right|^2\le \frac{C_\eta|\partial\Omega|}{\delta},
$$
where the constant $C_\eta>0$ depends only on the regularization function $\eta$.
\end{lemma}
In particular, $\sqrt{\iden_\Omega*\eta_\delta}\in H^1(\RR^3)$ whenever $|\partial\Omega|<\infty$. Importantly, this means that fermionic many-body localization (see \cite[Appendix A]{hainzl2009thermodynamic}) with respect to the function $\sqrt{\iden_\Omega*\eta_\delta}$ is meaningful.

We now return to the discussion of the lower bound on the indirect energy. For any $\vGamma\in\dm$ let
$$
\EC^\hbar(\vGamma)=\hbar^2\TC(\vGamma)+\Cc(\vGamma)-D(\rho_\vGamma)
$$
denote the full quantum indirect energy functional. 
Then it can be shown \cite{lewin2020local} using the IMS localization formula for the kinetic energy, and the smeared Graf--Schenker inequality for the indirect Coulomb energy that for any $\delta>0$ such that $0<\delta/\ell<1/C$,
\begin{multline}\label{indlowbd}
\EC^\hbar(\vGamma)\ge \left(1-\frac{C\delta}{\ell}\right) \int_{\SO(3)}\dd R  \dashint_{C_{\ell}}\dd \tau
\sum_{z\in\ZZ^3} \sum_{j=1}^{24} \EC^\hbar\bigl({\vGamma\bigr|}_{\sqrt{\iden_{\ell\tetra_j} * \eta_\delta(R\cdot - \ell z - \tau)}}\bigr)\\
 -\frac{C}{\ell} \int_{\RR^3} ((1+\hbar^2\delta^{-1})\rho_\vGamma + \delta^3\rho_\vGamma^2 )
\end{multline}
for some universal constant $C>0$. This immediately implies the following decoupling lower bound
on the indirect energy $E(\rho)$. 

\begin{theorem}[Decoupling lower bound]\label{lboundthm}
There is a universal constant $C>0$ such that for any $\ell>0$ and $\delta>0$ such that $0<\delta/\ell<1/C$ and $\rho\in L^1(\RR^3,\RR_+)$ with $\grad\sqrt{\rho}\in L^2(\RR^3)$ the bound
\begin{multline*}
E(\rho)\ge \left(1-\frac{C\delta}{\ell}\right) \int_{\SO(3)}\dd R  \dashint_{C_{\ell}}\dd \tau
\sum_{z\in\ZZ^3} \sum_{j=1}^{24} E\bigl(\iden_{\ell\tetra_j} * \eta_\delta(R\cdot - \ell z - \tau)\rho\bigr)\\
 -\frac{C}{\ell} \int_{\RR^3} ((1+\hbar^2\delta^{-1})\rho + \delta^3\rho^2 )
\end{multline*}
holds true. In particular, 
\begin{equation}\label{lboundmv}
E(\rho)\ge \left(1-\frac{C\delta}{\ell}\right)
\sum_{z\in\ZZ^3} \sum_{j=1}^{24} E\bigl(\iden_{\ell\tetra_j} * \eta_\delta(R\cdot - \ell z - \tau)\rho\bigr)
 -\frac{C}{\ell} \int_{\RR^3} ((1+\hbar^2\delta^{-1})\rho + \delta^3\rho^2 )
\end{equation}
for some isometry $(R,\tau)\in\SO(3)\times C_\ell$. 
\end{theorem}

By absorbing the isometry $(R,\tau)$ into $\rho$ in \cref{lboundmv}, we obtain plainly that the total energy is decoupled into the sum of the energies of ``lumps'' $\iden_{\ell\tetra_j} * \eta_\delta(\cdot - \ell z)\rho$ up to an error. The neighboring ``lumps'' in general will overlap. 

\subsection{\emph{A priori} bounds on the energy per volume} 
First, we consider the quantity $f_{\Omega,\eta_\delta,\zeta}(a)$ defined in \cref{fdef}.

\begin{proposition}\label{basicbdprop}
The following a priori bounds hold true.
\begin{itemize}
\item[(i)] For every $0<\epsilon\le 1/15$ and every inhomogeneity $\zeta$, there holds
\begin{multline*}
f_{\Omega,\eta_\delta,\zeta}(a)\le (1+\epsilon) c_{\TF}\hbar^2  \frac{1}{|\Omega|} \int_{\RR^3} (\iden_{\Omega}*\eta_\delta)\zeta^{5/3}(\cdot-a) \\
+ \frac{38}{15}\frac{\hbar^2}{\epsilon} \frac{1}{|\Omega|} \int_{\RR^3} (\iden_{\Omega}*\eta_\delta)|\grad \sqrt{\zeta}|^2(\cdot-a)+ \frac{38}{15}\frac{\hbar^2}{\epsilon} \frac{1}{|\Omega|} \int_{\RR^3} |\grad \sqrt{\iden_{\Omega}*\eta_\delta}|^2\zeta(\cdot-a) .
\end{multline*}
\item[(ii)] For every $0<\epsilon\le 3/5$ and every inhomogeneity $\zeta$, there holds 
\begin{multline*}
f_{\Omega,\eta_\delta,\zeta}(a)\ge (1-\epsilon) c_{\TF}\hbar^2  \frac{1}{|\Omega|} \int_{\RR^3} (\iden_{\Omega}*\eta_\delta)\zeta^{5/3}(\cdot-a) -c_\LO \frac{1}{|\Omega|}\int_{\RR^3} (\iden_{\Omega}*\eta_\delta)\zeta^{4/3}(\cdot-a) \\
- \frac{20}{27}\frac{\hbar^2}{\epsilon} \frac{1}{|\Omega|} \int_{\RR^3} (\iden_{\Omega}*\eta_\delta)|\grad \sqrt{\zeta}|^2(\cdot-a) + \frac{20}{27}\frac{\hbar^2}{\epsilon} \frac{1}{|\Omega|} \int_{\RR^3} |\grad \sqrt{\iden_{\Omega}*\eta_\delta}|^2\zeta(\cdot-a) .
\end{multline*}
\end{itemize}
\end{proposition}
\begin{proof}
Using \cref{llsebound} we have for $0<\epsilon\le 1/15$ that
\begin{align*}
&f_{\Omega,\eta_\delta,\zeta}(a)=\frac{1}{|\Omega|} E((\iden_{\Omega}*\eta_\delta)\zeta(\cdot-a))\\
&\le (1+\epsilon) c_{\TF}\hbar^2 \frac{1}{|\Omega|} \int_{\RR^3} (\iden_{\Omega}*\eta_\delta)^{5/3}\zeta^{5/3}(\cdot-a) + \frac{1}{|\Omega|} \frac{19}{15} \frac{\hbar^2}{\epsilon} \int_{\RR^3} |\grad (\sqrt{\iden_{\Omega}*\eta_\delta}\sqrt{\zeta(\cdot-a)})|^2\\
&\le (1+\epsilon) c_{\TF}\hbar^2  \frac{1}{|\Omega|} \int_{\RR^3} (\iden_{\Omega}*\eta_\delta)\zeta^{5/3}(\cdot-a) + \frac{38}{15}\frac{\hbar^2}{\epsilon} \frac{1}{|\Omega|} \int_{\RR^3} |\grad \sqrt{\iden_{\Omega}*\eta_\delta}|^2\zeta(\cdot-a) \\
&+ \frac{38}{15}\frac{\hbar^2}{\epsilon} \frac{1}{|\Omega|} \int_{\RR^3} (\iden_{\Omega}*\eta_\delta)|\grad \sqrt{\zeta}|^2(\cdot-a),
\end{align*}
where in the second step we used $0\le \iden_{\Omega}*\eta_\delta\le 1$ and elementary inequality
\begin{equation}\label{gradprodbd}
|\grad (\sqrt{\iden_{\Omega}*\eta_\delta}\sqrt{\zeta})|^2\le 2|\grad\sqrt{\iden_{\Omega}*\eta_\delta}|^2 \zeta + 2(\iden_{\Omega}*\eta_\delta)|\grad\sqrt{\zeta}|^2.
\end{equation}

For (ii), we have by the semiclassical Lieb--Thirring inequality (\cref{ltsemi}) and the Lieb--Oxford bound that
\begin{align*}
&f_{\Omega,\eta_\delta,\zeta}(a)\ge (1-\epsilon)c_\TF\hbar^2 \frac{1}{|\Omega|}\int_{\RR^3} (\iden_{\Omega}*\eta_\delta)^{5/3}\zeta^{5/3}(\cdot-a)\\
&- \frac{10}{27}\frac{\hbar^2}{\epsilon}\frac{1}{|\Omega|}\int_{\RR^3}|\grad (\sqrt{\iden_{\Omega}*\eta_\delta}\sqrt{\zeta(\cdot-a)})|^2
-c_\LO \frac{1}{|\Omega|}\int_{\RR^3} (\iden_{\Omega}*\eta_\delta)^{4/3}\zeta^{4/3}(\cdot-a)\\
&\ge (1-\epsilon)c_\TF\hbar^2 \frac{1}{|\Omega|}\int_{\RR^3} (\iden_{\Omega}*\eta_\delta)^{5/3}\zeta^{5/3}(\cdot-a) - \frac{20}{27}\frac{\hbar^2}{\epsilon}\frac{1}{|\Omega|}\int_{\RR^3}|\grad \sqrt{\iden_{\Omega}*\eta_\delta}|^2\zeta(\cdot-a)\\
&- \frac{20}{27}\frac{\hbar^2}{\epsilon}\frac{1}{|\Omega|}\int_{\RR^3}(\iden_{\Omega}*\eta_\delta)|\grad \sqrt{\zeta}|^2(\cdot-a)-c_\LO \frac{1}{|\Omega|}\int_{\RR^3} (\iden_{\Omega}*\eta_\delta)\zeta^{4/3}(\cdot-a)
\end{align*}
whenever $0<\epsilon\le 3/5$.
Here, in the second step we used \cref{gradprodbd} again.
\end{proof}

Averaging $f_{\Omega,\eta_\delta,\zeta}(a)$ over $a\in\RR^3$, and using \cref{indgradlemma} we obtain 
\begin{corollary}\label{fabsapriori}
For every inhomogeneity $\sqrt{\zeta}\in H_\per^1(\LC)$, the bound
$$
\dashint_{\RR^3} |f_{\Omega,\eta_\delta,\zeta}(a)|\,\dd a \le C \dashint_{\RR^3} (\hbar^2\zeta^{5/3} + \zeta^{4/3})
  + C\hbar^2 \dashint_{\RR^3}|\grad\sqrt{\zeta}|^2 + \frac{C\hbar^2|\partial\Omega|}{|\Omega|\delta}\dashint_{\RR^3}\zeta 
$$
holds true for some universal constant $C>0$ that only depends on $\eta$.
\end{corollary}

Recall the definition \cref{evoldef} of the quantum indirect energy per finite volume that we now simply denote as $e_{\Omega,\eta_\delta}(\zeta)$.

\begin{corollary}\label{aprioricoro}
For every inhomogeneity $\sqrt{\zeta}\in H_\per^1(\LC)$, the bound
$$
e_{\Omega,\eta_\delta}(\zeta) \le  (1+\epsilon) c_{\TF}\hbar^2 \dashint_{\RR^3}\zeta^{5/3}
+ \frac{38}{15}\frac{\hbar^2}{\epsilon} \dashint_{\RR^3}|\grad \sqrt{\zeta}|^2 + \frac{C_\eta\hbar^2}{\epsilon}\frac{|\partial\Omega|}{|\Omega|\delta}\dashint_{\RR^3}\zeta
$$
for any $0<\epsilon\le 1/15$, and the bound
$$
e_{\Omega,\eta_\delta}(\zeta) \ge (1-\epsilon)c_\TF\hbar^2\dashint_{\RR^3}\zeta^{5/3} -c_\LO\dashint_{\RR^3}\zeta^{4/3} - \frac{20}{27}\frac{\hbar^2}{\epsilon}\dashint_{\RR^3}|\grad \sqrt{\zeta}|^2-\frac{C_\eta\hbar^2}{\epsilon}\frac{|\partial\Omega|}{|\Omega|\delta}\dashint_{\RR^3}\zeta
$$
for any $0<\epsilon\le 3/5$ hold true. The universal constant $C_\eta>0$ only depends on $\eta$.
\end{corollary}

By taking $\Omega=\ell\tetra$, the prefactor of the last terms are proportional to $(\ell\delta)^{-1}$ that disappear in the thermodynamic limit, which yields the bounds of
\cref{qnuegrmrk} (ii), once the existence of the limit is proved.

\subsection{An improved spatial decoupling upper bound}\label{impdecoupsec}
In this section, we propose an estimate which slightly improves the spatial decoupling upper bound introduced in \cite{lewin2020local}.
This bound involves shrinking the smeared set $\iden_{\ell \tetra}*\eta_\delta$ so that its support is strictly contained in $\ell \tetra$, while keeping the smearing scale $\delta$ fixed. This is because in contrast to lower bound of \cref{lboundthm}, it would be much more difficult to localize the density in an overlapping manner for an upper bound.
Let $\delta\le\ell/2$. 
Using the Lions--Titchmarsh convolution theorem,
$$
\supp(\iden_{(\ell-\delta)\tetra} * \eta_\delta)=\ol{(\ell-\delta)\tetra} + (\supp\eta)_{b\delta},
$$
where we used our hypothesis that $\eta$ is radial so $\supp\eta$ is a closed ball, in particular convex. Here, $(\supp\eta)_{b\delta}$ denotes the dilation of the set $\supp\eta\subset B_1$ by $b\delta$. 
We also used the fact that $\supp(\iden_{(\ell-\delta)\tetra} * \eta_\delta)$ is convex, which can be seen from \cref{smearcases}.

So in fact we have $\supp(\iden_{(\ell-\delta)\tetra} * \eta_\delta)\subset \ell \tetra$. 
More generally, it is also clear from the above that decreasing the parameter $\delta$ to $c\delta$ strictly increases the supports as
\begin{equation}\label{suppinc}
\supp(\iden_{(\ell-\delta)\tetra} * \eta_\delta)\Subset\supp(\iden_{(\ell-c\delta)\tetra} * \eta_{c\delta})
\end{equation}
where $0<c<1$ and the symbol $A\Subset B$ means that $A\subset K\subset B$ for some compact set $K\subset\RR^3$. In fact, the distance from the boundary of $\supp(\iden_{(\ell-\delta)\tetra} * \eta_\delta)$ to the boundary of $\supp(\iden_{(\ell-c\delta)\tetra} * \eta_{c\delta})$ is $\ell-c\delta + cb\delta - (\ell -\delta+b\delta)=(1-c)(1-b)\delta>0$.

Our incomplete partition of unity reads
\begin{equation}\label{ipou}
\dashint_{C_\ell} \dd\tau \Bigg[ \sum_{z\in\ZZ^3} \sum_{j=1}^{24} \iden_{\ell T_j(1-\delta/\ell)\tetra} * \eta_\delta(\cdot - \ell z - \tau) + \Upsilon_{\ell,\delta}(\cdot-\tau)\Bigg]\equiv 1,
\end{equation}
where we have introduced the ``skeleton function''
$$
\Upsilon_{\ell,\delta}=\frac{1-(1-\delta/\ell)^3}{1-(1-\delta/(2\ell))^3}\Bigg(1-\sum_{z\in\ZZ^3} \sum_{j=1}^{24} \iden_{\ell T_j(1-\delta/(2\ell))\tetra} * \eta_{\delta/2}(\cdot-\ell z)\Bigg).
$$
Our point here is that now the smeared indicators of the first sum in \cref{ipou} are no longer multiplied by the normalization factor $(1-\delta/\ell)^{-3}$ as in \cite{lewin2020local}, because the skeleton function absorbs these weights. 
Clearly, $\Upsilon_{\ell,\delta}$ is a nonnegative, smooth $(\ell\ZZ^3)$-periodic function with disjoint support from all the other terms in \cref{ipou}. In fact, the support of 
$\iden_{\ell T_j(1-\delta/(2\ell))\tetra} * \eta_{\delta/2}$ strictly contains the support of 
$\iden_{\ell T_j(1-\delta/\ell)\tetra} * \eta_\delta$ according to \cref{suppinc}.
The mean value of the skeleton function is easily seen to be
$$
\dashint_{C_\ell} \Upsilon_{\ell,\delta} = 1 - (1-\delta/\ell)^3,
$$
which in turn verifies relation \cref{ipou}. Note also that this tends to zero in the limit $\delta/\ell\to 0$. According to \cref{indgradlemma}, the mean kinetic energy of  $\Upsilon_{\ell,\delta}$ is bounded as
$$
\dashint_{C_\ell} |\grad\sqrt{\Upsilon_{\ell,\delta}}|^2 \le \frac{C_\eta}{\ell\delta},
$$
which vanishes in the limit $\ell\delta\to\infty$. It turns out when decoupling the energy according to \cref{ipou}, all the contributions coming from the skeleton function $\Upsilon_{\ell,\delta}$
are negligible in the thermodynamic limit.

\begin{theorem}[Improved decoupling upper bound]\label{impdecup}
Let $0<\delta\le \ell/2$ and $0<\alpha<1/2$. Then for any $\rho\in L^1(\RR^3,\RR_+)$ with $\grad\sqrt{\rho}\in L^2(\RR^3)$ the bounds
\begin{multline*}
E(\rho)\le  \dashint_{1-\alpha}^{1+\alpha}\frac{\dd t}{t^4} \int_{\SO(3)}\dd R \dashint_{C_{t\ell}}\dd\tau \sum_{z\in\ZZ^3}  \sum_{j=1}^{24} 
E\bigl(\iden_{t\ell T_j(1-\delta/\ell)\tetra} * \eta_\delta(R\cdot - t\ell z - \tau)\rho\bigr)\\
+\frac{C\hbar^2\delta}{\ell} \int_{\RR^3}\rho^{5/3} +  \frac{C\hbar^2}{\ell\delta} \int_{\RR^3} \rho + \frac{C\hbar^2\delta}{\ell}\int_{\RR^3} |\grad\sqrt{\rho}|^2
+C\delta^2\log(\alpha^{-1}) \int_{\RR^3} \rho^2
\end{multline*}
and 
\begin{multline*}
T(\rho)\le \dashint_{C_{\ell}}\dd\tau \sum_{z\in\ZZ^3}  \sum_{j=1}^{24} T\bigl(\iden_{\ell T_j(1-\delta/\ell)\tetra} * \eta_\delta(\cdot - \ell z - \tau)\rho\bigr)\\
+\frac{C\hbar^2\delta}{\ell} \int_{\RR^3}\rho^{5/3} +  \frac{C\hbar^2}{\ell\delta} \int_{\RR^3} \rho + \frac{C\hbar^2\delta}{\ell}\int_{\RR^3} |\grad\sqrt{\rho}|^2
\end{multline*}
hold true for some universal constant $C>0$.
\end{theorem}
Here and henceforth, the notation 
$$
\dashint_{1-\alpha}^{1+\alpha}f(t)\frac{\dd t}{t^4} =\left(\int_{1-\alpha}^{1+\alpha}\frac{\dd t}{t^4} \right)^{-1} \int_{1-\alpha}^{1+\alpha}f(t) \frac{\dd t}{t^4} 
$$
is used.
\begin{proof}
Using our incomplete partition of unity \cref{ipou}, we may write
\begin{equation}\label{pourho}
\rho(x)=\dashint_{C_\ell} \Bigg[ \sum_{z\in\ZZ^3} \sum_{j=1}^{24} \iden_{\ell T_j(1-\delta/\ell)\tetra} * \eta_\delta(x - \ell z - \tau) + \Upsilon_{\ell,\delta}(x-\tau)\Bigg]\rho(x)\,\dd\tau \equiv 1.
\end{equation}
For any $\tau\in C_{\ell}$, $z\in\ZZ^3$ and $j=1,\ldots,24$ let $\Gamma_{\tau,z,j}$ be an optimizer of $F_\LL(\iden_{\ell T_j(1-\delta/\ell)\tetra} * \eta_\delta(\cdot - \ell z - \tau)\rho)$.
Also, let $\Gamma_{\tau,\Upsilon}$ be an optimizer of $F_\LL(\Upsilon_{\ell,\delta}(\cdot-\tau)\rho)$.
Next, we take 
$$
\Gamma_\tau=\bigotimes_{z\in\ZZ^3} \bigotimes_{j=1}^{24} \Gamma_{\tau,z,j} \otimes \Gamma_{\tau,\Upsilon}.
$$
All of the states $\Gamma_{\tau,z,j}$ and $\Gamma_{\tau,\Upsilon}$ have disjoint support for fixed $\tau$, so we may antisymmetrize $\Gamma_\tau$, which we denote by $\Gamma_{\tau,a}$.
This state has $\TC(\Gamma_{\tau,a})+\Cc(\Gamma_{\tau,a})=\TC(\Gamma_\tau)+\Cc(\Gamma_\tau)$ and $\rho_{\Gamma_{\tau,a}}=\rho_{\Gamma_\tau}$. 
Finally, using $\Gamma=\dashint_{C_{\ell}}\Gamma_{\tau,a}\dd\tau$ as a trial state, which has $\rho_\Gamma=\rho$, we obtain
\begin{align*}
&F_\LL(\rho)\le \dashint_{C_{\ell}}\dd\tau \sum_{z\in\ZZ^3} \sum_{j=1}^{24} F_\LL\bigl(\iden_{\ell T_j(1-\delta/\ell)\tetra} * \eta_\delta(\cdot - \ell z - \tau)\rho\bigr)+\dashint_{C_{\ell}} F_\LL\bigl(\Upsilon_{\ell,\delta}(\cdot - \tau)\rho\bigr)\, \dd\tau\\
&+ \dashint_{C_{\ell}} \dd\tau \sum_{\substack{(z,j),(z',j')\in\ZZ^3\times\{1,\ldots,24\}\\(z,j)\neq(z',j')}} 
D\left(\iden_{\ell T_j(1-\delta/\ell)\tetra} * \eta_\delta(\cdot - \ell z - \tau)\rho, \iden_{\ell T_{j'}(1-\delta/\ell)\tetra} * \eta_\delta(\cdot - \ell z' - \tau)\rho\right)\\
&+ 2\dashint_{C_{\ell}} \dd\tau \sum_{z\in\ZZ^3} 
\sum_{j=1}^{24} D\left(\iden_{\ell T_j(1-\delta/\ell)\tetra} * \eta_\delta(\cdot - \ell z - \tau)\rho, \Upsilon_{\ell,\delta}(\cdot-\tau)\rho\right).
\end{align*}
This implies
\begin{align*}
&E(\rho)\le \dashint_{C_{\ell}}\dd\tau \sum_{z\in\ZZ^3} \sum_{j=1}^{24} E\bigl(\iden_{\ell T_j(1-\delta/\ell)\tetra} * \eta_\delta(\cdot - \ell z - \tau)\rho\bigr)+\dashint_{C_{\ell}} E\bigl(\Upsilon_{\ell,\delta}(\cdot - \tau)\rho\bigr)\, \dd\tau\\
&-D(\rho)+\dashint_{C_{\ell}}\dd\tau \sum_{z\in\ZZ^3} \sum_{j=1}^{24} D\bigl(\iden_{\ell T_j(1-\delta/\ell)\tetra} * \eta_\delta(\cdot - \ell z - \tau)\rho\bigr) + \dashint_{C_{\ell}} D\bigl(\Upsilon_{\ell,\delta}(\cdot - \tau)\rho\bigr)\, \dd\tau\\
&+ \dashint_{C_{\ell}} \dd\tau \sum_{(z,j)\neq(z',j')}  D\left(\iden_{\ell T_j(1-\delta/\ell)\tetra} * \eta_\delta(\cdot - \ell z - \tau)\rho, \iden_{\ell T_{j'}(1-\delta/\ell)\tetra} * \eta_\delta(\cdot - \ell z' - \tau)\rho\right)\\
&+ 2\dashint_{C_{\ell}} \dd\tau \sum_{z\in\ZZ^3} 
\sum_{j=1}^{24} D\left(\iden_{\ell T_j(1-\delta/\ell)\tetra} * \eta_\delta(\cdot - \ell z - \tau)\rho, \Upsilon_{\ell,\delta}(\cdot-\tau)\rho\right).
\end{align*}
By inserting \cref{pourho} into $-D(\rho)$ as
\begin{multline*}
-D(\rho)=D(\rho)-2D(\rho,\rho)\\
=D(\rho)-2\dashint_{C_\ell}\dd\tau D\left(\rho,\sum_{z\in\ZZ^3} \sum_{j=1}^{24} \iden_{\ell T_j(1-\delta/\ell)\tetra} * \eta_\delta(\cdot - \ell z - \tau)\rho + \Upsilon_{\ell,\delta}(\cdot-\tau)\rho\right),
\end{multline*}
we see that the various direct terms give in total
$$
\dashint_{C_\ell} D\left(  \sum_{z\in\ZZ^3} \sum_{j=1}^{24} \iden_{\ell T_j(1-\delta/\ell)\tetra} * \eta_\delta(\cdot - \ell z - \tau)\rho + \Upsilon_{\ell,\delta}(\cdot-\tau)\rho -\rho\right)\,\dd\tau=:(*).
$$

First, we estimate the indirect energy error coming from the $\Upsilon_{\ell,\delta}$ term. Using \cref{llsebound} and the properties of $\Upsilon_{\ell,\delta}$, we have the bound
\begin{multline*}
\dashint_{C_\ell} E\bigl(\Upsilon_{\ell,\delta}(\cdot-\tau)\rho\bigr)\,\dd\tau \le C\hbar^2 \int_{\RR^3} \rho^{5/3} \dashint_{C_\ell}\Upsilon_{\ell,\delta}(\cdot-\tau)^{5/3}\,\dd\tau  \\
+C\hbar^2 \int_{\RR^3} \rho \dashint_{C_\ell}  |\grad\sqrt{\Upsilon_{\ell,\delta}}|^2 (\cdot-\tau)\,\dd\tau  + C\hbar^2\int_{\RR^3} |\grad\sqrt{\rho}|^2 \dashint_{C_\ell} \Upsilon_{\ell,\delta}(\cdot-\tau)\,\dd\tau\\
\le \frac{C\hbar^2\delta}{\ell} \int_{\RR^3}\rho^{5/3} +  \frac{C\hbar^2}{\ell\delta} \int_{\RR^3} \rho + \frac{C\hbar^2\delta}{\ell}\int_{\RR^3} |\grad\sqrt{\rho}|^2.
\end{multline*}

Finally, we estimate the direct term $(*)$. We proceed exactly as in \cite{lewin2020local}. Inserting
$$
\rho=\sum_{z\in\ZZ^3} \sum_{j=1}^{24} \iden_{\ell\tetra_j}(\cdot - \ell z - \tau)\rho \quad \text{(a.e.)}
$$
so that
$$
(*)=\dashint_{C_\ell} D\bigl( f(\cdot-\tau)\rho\bigr)\,\dd\tau,
$$
where the $(\ell\ZZ^3)$-periodic function $f:\RR^3\to\RR$ is given by
$$
f= \sum_{z\in\ZZ^3} \sum_{j=1}^{24} \left(\iden_{\ell\tetra_j}(\cdot - \ell z)- \iden_{\ell T_j(1-\delta/\ell)\tetra} * \eta_\delta(\cdot - \ell z)\right) - \Upsilon_{\ell,\delta}.
$$
Recall the representation from \cite{lewin2020local},
$$
\dashint_{C_\ell} D\bigl( f(\cdot-\tau)\rho\bigr)\,\dd\tau=2\pi \sum_{k\in (2\pi/\ell)\ZZ^3} \left|\dashint_{C_\ell} f(x)e^{-ik\cdot x}\,\dd x\right|^2 \int_{\RR^3} \frac{|\wh{\rho}(p)|^2}{|p-k|^2}\,\dd p,
$$
valid for any $(\ell\ZZ^3)$-periodic function $f$.

We obtain
$$
(*)=2\pi \sum_{k\in 2\pi\ZZ^3} \left|\int_{C_1} f_{\ell/\delta}(x)e^{-ik\cdot x}\,\dd x\right|^2 \int_{\RR^3} \frac{|\wh{\rho}(p)|^2}{|p-k/\ell|^2}\,\dd p.
$$
where we have set
$$
f_\epsilon=\sum_{j=1}^{24}\left(\iden_{\tetra_j}- \iden_{T_j(1-\epsilon)\tetra} * \eta_{\epsilon}\right) - \Upsilon_{\epsilon}.
$$
and
$$
\Upsilon_{\epsilon}=\frac{1-(1-\epsilon)^3}{1-(1-\epsilon/2)^3}\Bigg(1-\sum_{j=1}^{24} \iden_{T_j(1-\epsilon/2)\tetra} * \eta_{\epsilon/2}\Bigg).
$$
Here, $\int_{C_1} f_{\delta/\ell}=0$ and 
\begin{align*}
&\int_{C_1} f_{\ell/\delta}(x)e^{-ik\cdot x}\,\dd x\\
&= -(2\pi)^3\sum_{j=1}^{24} \wh{\iden_{T_j(1-\epsilon)\tetra}}(k) \wh{\eta}_1(\epsilon k) 
+ (2\pi)^3\frac{1-(1-\epsilon)^3}{1-(1-\epsilon/2)^3}\sum_{j=1}^{24} \wh{\iden_{T_j(1-\epsilon/2)\tetra}}(k)\wh{\eta}_1(\epsilon k/2).
\end{align*}
From this point on, the proof follows similarly to that of \cite[Proposition 1]{lewin2020local}. In conclusion, after
averaging over rotations and dilations $\ell\to t\ell$ and $\delta\to t\delta$ over $t\in(1-\alpha,1+\alpha)$ with weight $t^{-4}$, we obtain the bound
$$
\dashint_{1-\alpha}^{1+\alpha}\frac{\dd t}{t^4} \int_{\SO(3)}\dd R \dashint_{C_{t\ell}}\dd\tau\, D\bigl( f(R(\cdot-\tau))\rho\bigr)\le C\delta^2\log(\alpha^{-1}) \int_{\RR^3} \rho^2,
$$
which concludes the proof.
\end{proof}

\subsection{Convergence rate for tetrahedra}

In this section we study the thermodynamic convergence of the indirect energy
for dilated tetrahedra.

\begin{theorem}[Convergence rate for tetrahedra]\label{tetrathermo}
Fix an inhomogeneity $\zeta$ and a regularization function $\eta$. Let $\ul{e}_\NUEG(\zeta)$ and $\ol{e}_\NUEG(\zeta)$ denote the liminf and the limsup of $e_{\ell\tetra,\eta_\delta}(\zeta)$
as $\ell\delta\to\infty$ and $\delta/\ell\to 0$. 
\begin{itemize}
\item[(i)] For $\ell/\delta$ sufficiently large, there holds
\begin{align*}
&e_{\ell\tetra,\eta_{\delta}}(\zeta)\le \ul{e}_\NUEG(\zeta) +\frac{C}{\ell}\Bigg( (1+\hbar^2\delta^{-1}) \dashint_{\RR^3} \zeta + \delta^3\dashint_{\RR^3} \zeta^2 + \hbar^2\delta \dashint_{\RR^3} \zeta^{5/3} + \hbar^2\delta \dashint_{\RR^3} |\grad\sqrt{\zeta}|^2\Bigg)
\end{align*}
\item[(ii)] For $\ell/\delta$ sufficiently large and $0<\alpha<1/2$, 
\begin{multline*}
\ol{e}_{\NUEG}(\zeta)\le \dashint_{1-\alpha}^{1+\alpha}\frac{\dd t}{t^4}\, e_{t\ell\tetra,\eta_{t\delta}}(\zeta)\\
+\frac{C\hbar^2\delta}{\ell} \dashint_{\RR^3}\zeta^{5/3} +  \frac{C\hbar^2}{\ell\delta} \dashint_{\RR^3} \zeta + \frac{C\hbar^2\delta}{\ell}\dashint_{\RR^3} |\grad\sqrt{\zeta}|^2+ C\delta^2\log(\alpha^{-1}) \dashint_{\RR^3}\zeta^2
\end{multline*}
\item[(iii)] For $\ell$  sufficiently large, 
\begin{align*}
&e_{\ell\tetra,\eta_{\delta}}(\zeta)\ge \ol{e}_{\NUEG}(\zeta) -\frac{C\delta}{\ell}\dashint_{\RR^3} \left(\zeta +\zeta^2\right) -\frac{C}{\ell^{2/5}}\dashint_{\RR^3} \left(\zeta +\zeta^2\right) - \frac{C}{\ell^{4/5}}\dashint_{\RR^3} |\grad\sqrt{\zeta}|^2.
\end{align*}
\item[(iv)] The thermodynamic limit
$$
\lim_{\substack{\delta/\ell\to 0\\\ell\delta\to\infty}} e_{\ell\tetra,\eta_\delta}(\zeta)=e_\NUEG(\zeta)
$$
exists and is independent of the regularization function $\eta$.
\end{itemize}
In all the above bounds, the generic constant $C>0$ depends only on the regularization function $\eta$ and is independent of $\ell$, $\delta$, $\zeta$ and $\alpha$. 
\end{theorem}

Before proving the above theorem, we derive a bound comparing the energy of a big tetrahedron to a small tetrahedron. This result can be thought of as the quantum analog of \cref{dyadmono}.
\begin{lemma}[Monotonicity estimate]\label{monolemm}
Fix an inhomogeneity $\zeta$ and a regularization function $\eta$. For any $\ell'\gg\ell>0$, $\delta',\delta>0$ such that $\delta/\ell\le 1/C$, the bound
\begin{equation}\label{tetracomp}
\begin{multlined}
e_{\ell'\tetra,\eta_{\delta'}}(\zeta)\ge \Bigg( 1-C\sigma \frac{\delta}{\ell} -C\sigma\frac{\ell+\delta+\delta'}{\ell'}\Bigg) e_{\ell\tetra,\eta_{\delta}}(\zeta)
- C \frac{\ell+\delta+\delta'}{\ell'} \dashint_{\RR^3} \zeta^{4/3}\\
-\frac{C}{\ell}\Bigg( (1+\hbar^2\delta^{-1}) \dashint_{\RR^3} \zeta + \delta^3\dashint_{\RR^3} \zeta^2 \Bigg)
\end{multlined}
\end{equation}
holds true, where
$$
\sigma=\begin{cases}
0 & \text{if} \; e_{\ell\tetra,\eta_\delta}(\zeta)\le 0 \\
1 & \text{otherwise}
\end{cases}
$$
\end{lemma}
\begin{proof}
Let $\ell\ll \ell'$ sufficiently large, $\delta/\ell\le 1/C$ and write 
\begin{align*}
&e_{\ell'\tetra,\eta_{\delta'}}(\zeta)=\frac{1}{|\ell'\tetra|}\int_{\SO(3)}\dd Q \dashint_{\RR^3}\dd a \, E\bigl((\iden_{\ell'\tetra}*\eta_{\delta'})\zeta(Q(\cdot-a))\bigr)\\
&\ge \frac{1}{|\ell'\tetra|} \left(1-\frac{C\delta}{\ell}\right) \int_{\SO(3)}\dd Q \dashint_{\RR^3} \dd a  \int_{\SO(3)}\dd R  \dashint_{C_{\ell}}\dd\tau  \\
&\quad\quad\times \sum_{z\in\ZZ^3} \sum_{j=1}^{24} E\bigl(\iden_{\ell\tetra_j}*\eta_\delta(R\cdot - \ell z - \tau)(\iden_{\ell'\tetra}*\eta_{\delta'})\zeta(Q(\cdot -a))\bigr)\\
&-\frac{C}{\ell|\ell'\tetra|} \int_{\SO(3)}\dd Q \dashint_{\RR^3}\dd a \int_{\RR^3} \Bigl( (1+\hbar^2\delta^{-1})(\iden_{\ell'\tetra}*\eta_{\delta'}) \zeta(Q(\cdot-a)) \\
&\quad\quad+ \delta^3(\iden_{\ell'\tetra}*\eta_{\delta'})^2\zeta(Q(\cdot-a))^2 \Bigr)\\
&=:\mathrm{(I)} + \mathrm{(II)}
\end{align*}
Using $0\le \iden_{\ell'\tetra}*\eta_{\delta'}\le 1$, and evaluating the mean value with respect to $a$, the terms in (II) may be estimated as
\begin{align*}
\mathrm{(II)}\ge -\frac{C}{\ell}\Bigg( (1+\hbar^2\delta^{-1}) \dashint_{\RR^3} \zeta + \delta^3\dashint_{\RR^3} \zeta^2 \Bigg) 
\end{align*}

Next, we consider (I). Define the index set
$$
J:=\left\{ \begin{aligned}(z,j)\in \ZZ^3\times\{1,\ldots,24\} : &\supp( (\iden_{\ell\tetra_j}*\eta_\delta)(R\cdot - \ell z - \tau)) \cap \supp(\iden_{\ell'\tetra}*\eta_{\delta'})\neq \emptyset,\\
&\text{for some}\;R\in\SO(3)\;\text{and}\; \tau\in C_\ell \end{aligned} \right\},
$$
which collects all the small smeared tetrahedra that possibly intersect with the big smeared tetrahedron $\supp(\iden_{\ell'\tetra}*\eta_{\delta'})$. Furthermore, define
$$
J_0:=\left\{ \begin{aligned}(z,j)\in \ZZ^3\times\{1,\ldots,24\} : &\supp( (\iden_{\ell\tetra_j}*\eta_\delta)(R\cdot - \ell z - \tau)) \subset (\ell'-\delta')\tetra,\\
&\text{for all}\;R\in\SO(3)\;\text{and}\; \tau\in C_\ell \end{aligned} \right\},
$$
which contains small smeared tetrahedra which are well inside $\ell'\tetra$.

According to $J_0$ and $J\setminus J_0$, we may decompose the sum in (I) to (Ia) and (Ib), respectively. 
To estimate (Ia), note that for $(z,j)\in J_0$, we simply have
$$
(\iden_{\ell\tetra_j}*\eta_\delta)(R\cdot - \ell z - \tau)(\iden_{\ell'\tetra}*\eta_{\delta'})=(\iden_{\ell\tetra_j}*\eta_\delta)(R\cdot - \ell z - \tau),
$$
since $\iden_{\ell'\tetra}*\eta_{\delta'}\equiv 1$ on $(\ell'-\delta')\tetra$. Therefore,
\begin{multline*}
\mathrm{(Ia)}= \frac{1}{|\ell'\tetra|} \left(1-\frac{C\delta}{\ell}\right) \int_{\SO(3)}\dd Q  \dashint_{\RR^3} \dd a \int_{\SO(3)}\dd R\dashint_{C_{\ell}}\dd\tau \\
\times \sum_{(z,j)\in J_0} E\bigl(\iden_{\ell\tetra_j}*\eta_\delta(R\cdot - \ell z - \tau)\zeta(Q(\cdot-a))\bigr)
\end{multline*}
Using the definition \cref{evoldef}, this can be written as
\begin{align*}
\mathrm{(Ia)}&= \frac{1}{|\ell'\tetra|} \left(1-\frac{C\delta}{\ell}\right)   \int_{\SO(3)}\,\dd R\dashint_{C_{\ell}}\dd\tau \sum_{(z,j)\in J_0} |\ell\tetra|e_{R^\top\ell\tetra_j+\ell z+\tau,\eta_\delta}(\zeta)
\end{align*}
Recalling that $\Omega\mapsto e_{\Omega,\eta_\delta}(\zeta)$ is isometry-invariant, we find
\begin{align*}
\mathrm{(Ia)}&= \frac{M_\ell^0}{|\ell'\tetra|} \left(1-\frac{C\delta}{\ell}\right)  e_{\ell\tetra,\eta_\delta}(\zeta)
\end{align*}
where we have set $M_\ell^0=|\ell\tetra||J_0|$.

Next, we estimate the contributions from tetrahedra close to the boundary. 
Using the Lieb--Oxford inequality (\cref{lobound}) we find
\begin{multline*}
\mathrm{(Ib)}\ge \frac{C}{|\ell'\tetra|} \left(1-\frac{C\delta}{\ell}\right) \int_{\SO(3)}\dd Q \dashint_{\RR^3} \dd a \int_{\SO(3)}\dd R  \dashint_{C_{\ell}}\dd\tau \\
\times\sum_{(z,j)\in J\setminus J_0}\Bigg[
-c_\LO \int_{\RR^3} \iden_{\ell\tetra_j}*\eta_\delta(R\cdot - \ell z - \tau)^{4/3}(\iden_{\ell'\tetra}*\eta_{\delta'})^{4/3} \zeta(Q(\cdot -a))^{4/3}\Bigg] 
\end{multline*}
This can be further bounded as
$$
\mathrm{(Ib)}\ge - \frac{M_\ell^\partial}{|\ell'\tetra|} \left(1-\frac{C\delta}{\ell}\right)  c_\LO \dashint_{\RR^3} \zeta^{4/3},
$$
where we have set $M_\ell^\partial=|\ell\tetra||J\setminus J_0|$. 

Let us bound the volume ratios $M_\ell^0/|\ell'\tetra|$ and $M_\ell^\partial/|\ell'\tetra|$. If $e_{\ell\tetra,\eta_\delta}(\zeta)\le 0$, then may use the trivial bound $M_\ell^0/|\ell'\tetra|\le 1$. Recall that the index set $J\setminus J_0$ collects the small tetrahedra close to the boundary of $\ell'\tetra$ and are at a distance $\OC(\ell+\delta+\delta')$ from $\partial(\ell'\tetra)$. Since the surface area of $\partial(\ell'\tetra)$ is $\OC(\ell'^2)$, the tetrahedra close to the boundary fill a volume of $M_\ell^\partial=\OC(\ell'^2(\ell+\delta+\delta'))$. Therefore,
\begin{equation}\label{volbubound}
\frac{M_\ell^0}{|\ell'\tetra|}\ge \frac{|\ell'\tetra|-C\ell'^2(\ell+\delta+\delta')}{|\ell'\tetra|}=1-C\frac{\ell+\delta+\delta'}{\ell'}
\end{equation}
and 
\begin{equation}\label{volbdbound}
\frac{M_\ell^\partial}{|\ell'\tetra|}\le C \frac{\ell+\delta+\delta'}{\ell'}.
\end{equation}
With these at hand, the stated estimate follows.
\end{proof}

\begin{proof}[Proof of \cref{tetrathermo}]
\noindent\textbf{Proof of (i).}
We split the l.h.s. of \cref{tetracomp} according to 
\begin{align*}
e_{\ell'\tetra,\eta_{\delta'}}(\zeta)&=(1-C\sigma\delta/\ell) e_{\ell'\tetra,\eta_{\delta'}}(\zeta) +  \frac{C\sigma\delta}{\ell} e_{\ell'\tetra,\eta_{\delta'}}(\zeta)\\
&\le (1-C\sigma\delta/\ell) e_{\ell'\tetra,\eta_{\delta'}}(\zeta) + \frac{C\hbar^2\sigma\delta}{\ell} \Bigg[ \dashint_{\RR^3} \zeta^{5/3} + \dashint_{\RR^3} |\grad\sqrt{\zeta}|^2 +  \frac{1}{\ell'\delta'}\dashint_{\RR^3} \zeta \Bigg]
\end{align*}
where we used \cref{basicbdprop} (i) with $\epsilon=1$. Rearranging our bound, we find
\begin{multline*}
\Bigg( 1-C\sigma \frac{\delta}{\ell} -C\sigma\frac{\ell+\delta+\delta'}{\ell'}\Bigg) e_{\ell\tetra,\eta_{\delta}}(\zeta)\\
\le (1-C\sigma\delta/\ell) e_{\ell'\tetra,\eta_{\delta'}}(\zeta)
 +\frac{C}{\ell}\Bigg( (1+\hbar^2\delta^{-1}) \dashint_{\RR^3} \zeta + \delta^3\dashint_{\RR^3} \zeta^2 + \hbar^2\delta \dashint_{\RR^3} \zeta^{5/3} + \hbar^2\delta \dashint_{\RR^3} |\grad\sqrt{\zeta}|^2\Bigg) \\
 + C \frac{\ell+\delta+\delta'}{\ell'} \dashint_{\RR^3} \zeta^{4/3} + \frac{C\hbar^2\sigma\delta}{\ell}  \frac{1}{\ell'\delta'}\dashint_{\RR^3} \zeta
\end{multline*}
Now taking the liminf as $\ell'\to\infty$, and dividing by $(1-C\sigma\delta/\ell)$, we arrive at
\begin{align*}
&e_{\ell\tetra,\eta_{\delta}}(\zeta)\le \ul{e}_\UEG(\zeta) +\frac{C}{\ell}\Bigg( (1+\hbar^2\delta^{-1}) \dashint_{\RR^3} \zeta + \delta^3\dashint_{\RR^3} \zeta^2 + \hbar^2\delta \dashint_{\RR^3} \zeta^{5/3} + \hbar^2\delta \dashint_{\RR^3} |\grad\sqrt{\zeta}|^2\Bigg),
\end{align*}
which is what we wanted to show.

\vspace{.3em}
\noindent\textbf{Proof of (ii).} The proof of the lower bound is similar. Using \cref{impdecup}, we have for $0<\alpha<\frac{1}{2}$, 
\begin{equation}\label{etetraup}
\begin{aligned}
&e_{\ell'\tetra,\eta_{\delta'}}(\zeta)\le\frac{1}{|\ell'\tetra|}\int_{\SO(3)}\dd Q \dashint_{\RR^3}\dd a \dashint_{1-\alpha}^{1+\alpha}\frac{\dd t}{t^4} \int_{\SO(3)}\dd R 
\dashint_{C_{t(\ell+\delta)}}\dd\tau\\
	&\times \sum_{z\in\ZZ^3} \sum_{j=1}^{24} E\bigl((\iden_{t(\ell+\delta) T_j \beta\tetra} * \eta_{t \delta})(R\cdot - t(\ell+\delta) z - \tau)(\iden_{\ell'\tetra}*\eta_{\delta'})\zeta(Q(\cdot - a))\bigr)\\
	&+\Biggl[\frac{C\hbar^2\delta}{\ell} \dashint_{\RR^3}\zeta^{5/3} +  \frac{C\hbar^2}{\ell\delta} \dashint_{\RR^3} \zeta + \frac{C\hbar^2\delta}{\ell}\dashint_{\RR^3} |\grad\sqrt{\zeta}|^2 + C\delta^2\log(\alpha^{-1}) \dashint_{\RR^3}\zeta^2\Biggr]\\
    &=\mathrm{(I)} + \mathrm{(II)}
\end{aligned}
\end{equation}
where we used the notation $\beta=\ell/(\ell+\delta)$, where $\frac{2}{3}\le\beta<1$. Notice that the scales are chosen so that the small tetrahedra are of size $t\ell$ and smeared at scale $t\delta$.

Again, we split (I) according to $J_0$ and $J\setminus J_0$, which we define analogously to above.
First, we consider the summation over $J_0$ in (I), which we denote by (Ia).
Using $t(\ell+\delta) T_j(1-\delta/(\ell+\delta))\tetra=t(\ell+\delta)\beta R_j\tetra + t(\ell+\delta) z_j$, we may write
\begin{multline*}
(\iden_{t(\ell+\delta) T_j\beta\tetra} * \eta_{t\delta})(Rx - t(\ell+\delta) z - \tau)\\
=(\iden_{t\ell\tetra} * \eta_{t\delta})(R_j^\top Rx - R_j^\top(t(\ell+\delta) z + \tau + t(\ell+\delta) z_j))
\end{multline*}
We have
\begin{align*}
&\mathrm{(Ia)}=\frac{1}{|\ell'\tetra|} \dashint_{1-\alpha}^{1+\alpha}\frac{\dd t}{t^4}  \int_{\SO(3)}\dd Q \dashint_{\RR^3}\dd a \int_{\SO(3)}\dd R \dashint_{C_{t(\ell+\delta)}}\dd\tau\sum_{(z,j)\in J_0}\\
	&\times \dashint_{\RR^3}\int_{\SO(3)} E\bigl((\iden_{t\ell\tetra} * \eta_{t\delta})(R_j^\top R\cdot - R_j^\top(t(\ell+\delta) z + \tau + t(\ell+\delta) z_j))\zeta(Q(\cdot - a))\bigr)\\
&=\frac{1}{|\ell'\tetra|} \dashint_{1-\alpha}^{1+\alpha}\frac{\dd t}{t^4} |J_0|
	\int_{\SO(3)}\dd Q \dashint_{\RR^3}\dd a\,  E\bigl((\iden_{t\ell\tetra} * \eta_{t\delta})\zeta( Q(\cdot- a))\bigr)\\
	&=\dashint_{1-\alpha}^{1+\alpha}\frac{\dd t}{t^4} \frac{M_{t\ell}^0}{|\ell'\tetra|} e_{t\ell\tetra,\eta_{t\delta}}(\zeta)
\end{align*}
In the second equality, we used the isometry-invariance of the energy and \cref{permeanthm}.
Here, the volume ratio may be bounded similarly as before
$$
\frac{M_{t\ell}^0}{|\ell'\tetra|}\le 1+C\sigma\frac{t\ell+t\delta+\delta'}{\ell'}\le 1+C\sigma\frac{\ell+\delta+\delta'}{\ell'}.
$$
Next, we estimate the contributions near the boundary using \cref{llsebound},
\begin{multline*}
\mathrm{(Ib)}=\frac{1}{|\ell'\tetra|}\int_{\SO(3)}\dd Q \dashint_{\RR^3}\dd a \dashint_{1-\alpha}^{1+\alpha}\frac{\dd t}{t^4} \int_{\SO(3)}\dd R \dashint_{C_{t(\ell+\delta)}}\dd\tau
\sum_{(z,j)\in J\setminus J_0} \times\\
	\times E\bigl((\iden_{t(\ell+\delta) T_j \beta\tetra} * \eta_{t\delta})(R\cdot - t(\ell+\delta) z - \tau)(\iden_{\ell'\tetra}*\eta_{\delta'})\zeta(Q(\cdot - a))\bigr)\\
\le \hbar^2\frac{1}{|\ell'\tetra|}\int_{\SO(3)}\dd Q \dashint_{\RR^3}\dd a \dashint_{1-\alpha}^{1+\alpha}\frac{\dd t}{t^4} \int_{\SO(3)}\dd R \dashint_{C_{t(\ell+\delta)}}\dd\tau
\sum_{(z,j)\in J\setminus J_0} \times\\
	\times\Bigg[ C\int_{\RR^3}  (\iden_{t(\ell+\delta) T_j \beta\tetra} * \eta_{t\delta})(R\cdot - t(\ell+\delta) z - \tau)(\iden_{\ell'\tetra}*\eta_{\delta'})\zeta(Q(\cdot - a))^{5/3}\\
    + C \int_{\RR^3} \Bigl|\grad\sqrt{(\iden_{t(\ell+\delta) T_j \beta\tetra} * \eta_{t\delta})(R\cdot - t(\ell+\delta) z - \tau)(\iden_{\ell'\tetra}*\eta_{\delta'})\zeta(Q(\cdot - a))}\Bigr|^2 \Bigg]
\end{multline*}
Here, the gradient term may be bounded using the Cauchy--Schwarz inequality as
\begin{align*}
C\int_{\RR^3} \Bigl|\grad\sqrt{(\iden_{t(\ell+\delta) T_j \beta\tetra} * \eta_{t\delta})(R\cdot - t(\ell+\delta) z - \tau)}\Bigr|^2(\iden_{\ell'\tetra}*\eta_{\delta'})\zeta(Q(\cdot - a))\\
+C\int_{\RR^3} \Bigl|\grad\sqrt{\iden_{\ell'\tetra}*\eta_{\delta'}}\Bigr|^2
(\iden_{t(\ell+\delta) T_j \beta\tetra} * \eta_{t\delta})(R\cdot - t(\ell+\delta) z - \tau)\zeta(Q(\cdot - a))\\
+C\int_{\RR^3} \bigl|\grad\sqrt{\zeta(Q(\cdot - a))}\bigr|^2(\iden_{t(\ell+\delta) T_j \beta\tetra} * \eta_{t\delta})(R\cdot - t(\ell+\delta) z - \tau)(\iden_{\ell'\tetra}*\eta_{\delta'})
\end{align*}
After computing the averages with respect to $a$ and $Q$, we obtain by \cref{indgradlemma},
\begin{equation}\label{ebduest}
\mathrm{(Ib)}\le C\hbar^2 \frac{M_{t\ell}^\partial}{|\ell'\tetra|}\Bigg[ \dashint_{\RR^3} \zeta^{5/3} + \frac{1}{\ell\delta}\dashint_{\RR^3}\zeta + \dashint_{\RR^3} |\grad\sqrt{\zeta}|^2\Bigg] + \frac{C\hbar^2}{\ell'\delta'}\dashint_{\RR^3}\zeta.
\end{equation}
Combining this with the bound \cref{volbdbound}, we arrive at an estimate on (Ib) that vanishes in the thermodynamic limit $\ell'\delta'\to\infty$, $\delta'/\ell'\to 0$. 

Collecting our estimates and taking the limsup as $\ell'\delta'\to\infty$, $\delta'/\ell'\to 0$, we obtain
\begin{multline*}
\ol{e}_{\NUEG}(\zeta)\le \dashint_{1-\alpha}^{1+\alpha}\frac{\dd t}{t^4} e_{t\ell\tetra,\eta_{t\delta}}(\zeta)\\
+\frac{C\hbar^2\delta}{\ell} \dashint_{\RR^3}\zeta^{5/3} +  \frac{C\hbar^2}{\ell\delta} \dashint_{\RR^3} \zeta + \frac{C\hbar^2\delta}{\ell}\dashint_{\RR^3} |\grad\sqrt{\zeta}|^2+ C\delta^2\log(\alpha^{-1}) \dashint_{\RR^3}\zeta^2
\end{multline*}
as stated.

\vspace{.3em}
\noindent\textbf{Proof of (iii).} In order to get rid of the averaging in (ii), we use \cref{tetracomp} with $\ell\to t\ell$, $\delta\to t\delta$ and average over $t\in(1/2,3/2)$. We obtain
\begin{multline}\label{monoavg}
e_{\ell'\tetra,\eta_{\delta'}}(\zeta)\ge \Bigg( 1-C\sigma \frac{\delta}{\ell} -C\sigma\frac{\ell+\delta+\delta'}{\ell'}\Bigg) \dashint_{1-\alpha}^{1+\alpha}\frac{\dd t}{t^4} e_{t\ell\tetra,\eta_{t\delta}}(\zeta)\\
- C \frac{\ell+\delta+\delta'}{\ell'} \dashint_{\RR^3} \zeta^{4/3}-\frac{C}{\ell}\Bigg( (1+\hbar^2\delta^{-1}) \dashint_{\RR^3} \zeta + \delta^3\dashint_{\RR^3} \zeta^2 \Bigg).
\end{multline}
Using (ii), we may bound the $t$-average from below by 
$$
\ol{e}_{\NUEG}(\zeta) -\frac{C\hbar^2\delta}{\ell} \dashint_{\RR^3}\zeta^{5/3} -  \frac{C\hbar^2}{\ell\delta} \dashint_{\RR^3} \zeta - \frac{C\hbar^2\delta}{\ell}\dashint_{\RR^3} |\grad\sqrt{\zeta}|^2 -C\delta^2\dashint_{\RR^3}\zeta^2,
$$
to get after discarding the positive terms and rearranging,
\begin{align*}
&e_{\ell'\tetra,\eta_{\delta'}}(\zeta)\ge \ol{e}_{\NUEG}(\zeta) - C\frac{\ell+\delta+\delta'}{\ell'} \dashint_{\RR^3} \zeta^{4/3} \\
&-\frac{C}{\ell}\Biggl( (1+\hbar^2\delta^{-1}) \dashint_{\RR^3} \zeta + \delta^3\dashint_{\RR^3} \zeta^2 + \hbar^2\delta \dashint_{\RR^3} \zeta^{5/3} + \hbar^2\delta \dashint_{\RR^3} |\grad\sqrt{\zeta}|^2\Biggr) - C\delta^2 \dashint_{\RR^3}\zeta^2  \\
&\ge \ol{e}_{\NUEG}(\zeta) - C\frac{\ell+\delta+\delta'}{\ell'} \dashint_{\RR^3} \left(\zeta +\zeta^2\right) \\
&-C\frac{1+\delta^{-1}+\delta}{\ell} \dashint_{\RR^3} \zeta - C\left(\frac{\delta+\delta^3}{\ell}+\delta^2 \right)\dashint_{\RR^3} \zeta^2  -\frac{C\delta}{\ell} \dashint_{\RR^3} |\grad\sqrt{\zeta}|^2
\end{align*}
where in the last step we made use of $\zeta^\alpha\le \zeta+\zeta^2$ for $1\le\alpha\le 2$
Here, we choose $\delta=\ell^{-1/3}$ and then $\ell=(\ell')^{3/5}$ to obtain
\begin{align*}
&e_{\ell'\tetra,\eta_{\delta'}}(\zeta)\ge \ol{e}_{\NUEG}(\zeta) -\frac{C\delta'}{\ell'}\dashint_{\RR^3} \left(\zeta +\zeta^2\right) -\frac{C}{(\ell')^{2/5}}\dashint_{\RR^3} \left(\zeta +\zeta^2\right) - \frac{C}{(\ell')^{4/5}}\dashint_{\RR^3} |\grad\sqrt{\zeta}|^2
\end{align*}
for sufficiently large $\ell'$.

\vspace{.3em}
\noindent\textbf{Proof of (iv).} It only remains to explain why the limit is independent of the
regularization function $\eta$. Suppose that $\wt{\eta}$ is another regularization function. 
One can derive estimates similar to (i)--(iii) with finite-volume quantities involving $\wt{\eta}$ instead.
 This implies that the thermodynamic limit of $e_{\ell\tetra,\wt{\eta}_\delta}(\zeta)$ is the same as that of $e_{\ell\tetra,\eta_\delta}(\zeta)$.
\end{proof}

\subsection{The equivalence of the two definitions of $e_\NUEG$ for tetrahedra}
In a preliminary step towards general domain sequences, we prove that the thermodynamic limit of $\wt{e}_{\Omega,s}(\zeta)$ for a dilated reference tetrahedron exists, and is equal to the limit in \cref{tetrathermo}.

\begin{proposition}\label{tetraeq}
Fix an inhomogeneity $\zeta$ and a regularization function $\eta$. Then the limits
\begin{align*}
&\lim_{\substack{L_N\to\infty\\s_N/L_N\to 0}} \wt{e}_{L_N\tetra,s_N}(\zeta)=\lim_{\ell\to\infty} e_{\ell\tetra,\eta}(\zeta)=e_\NUEG(\zeta)
\end{align*}
exist and are independent of the sequences $L_N$ and $s_N$ and of $\eta$.
\end{proposition}
Notice that in the second limit, we simply took the constant sequence $\delta=1$ of smearing scale, which according to \cref{tetrathermo} converges to $e_\NUEG(\zeta)$.
\begin{proof}
Recall definition \cref{evoltrans}, which in our case reads
$$
\wt{e}_{L_N\tetra,s_N}(\zeta)=\int_{\SO(3)}\dd  R \dashint_{\RR^3}\dd a\, \inf_{\substack{\sqrt{\rho}\in H^1(\RR^3)\\ \zeta(R(\cdot-a))\iden_{(L_N-s_N)\tetra}\le \rho \le \zeta(R(\cdot-a))\iden_{(L_N+s_N)\tetra} }} \frac{E(\rho)}{|L_N\tetra|}
$$
For sufficiently large $N$, the pointwise bound
$$
\zeta(R(\cdot-a))\iden_{(L_N-s_N)\tetra}\le (\iden_{L_N\tetra}*\eta)\zeta(R(\cdot-a)) \le \zeta(R(\cdot-a))\iden_{(L_N+s_N)\tetra}
$$
holds true. In other words, $(\iden_{L_N\tetra}*\eta)\zeta(R(\cdot-a))$ is a competing density in the above infimum.
 We obtain that for $N$ sufficiently large
\begin{equation}\label{etranscomp}
\wt{e}_{L_N\tetra,s_N}(\zeta)\le e_{L_N\tetra,\eta}(\zeta).
\end{equation}

For the lower bound on $\wt{e}_{L_N\tetra,s_N}(\zeta)$ let $\vGamma$ be a state such that
$$
\zeta(R(\cdot-a))\iden_{(L_N-s_N)\tetra}\le \rho_{\vGamma} \le \zeta(R(\cdot-a))\iden_{(L_N+s_N)\tetra}.
$$
Using the lower bound \cref{indlowbd} with $\delta=1$ and dividing by the volume we get
\begin{multline*}
\frac{\EC(\vGamma)}{|L_N\tetra|}\ge \left(1-\frac{C}{\ell}\right) \int_{\SO(3)}\dd Q  \dashint_{C_{\ell}}\dd \tau
\sum_{z\in\ZZ^3} \sum_{j=1}^{24} \frac{\EC\bigl({\vGamma\bigr|}_{\sqrt{\iden_{\ell\tetra_j} * \eta(Q\cdot - \ell z - \tau)}}\bigr)}{|L_N\tetra|}\\
 -\frac{C}{\ell} \frac{1}{|L_N\tetra|}\int_{\RR^3} (\rho_\vGamma + \rho_\vGamma^2 ).
\end{multline*}
As before, we split the $(z,j)$-sum to interior-, and boundary parts and estimate the boundary contributions using the Lieb--Oxford bound. 
Using our constraint on $\rho_\vGamma$ for the last term, optimizing over $\vGamma$ and then averaging over $R$ and $a$ we find
\begin{multline*}
\wt{e}_{L_N\tetra,s_N}(\zeta)\ge \left(1-\frac{C}{\ell}\right) \left(1-C\frac{\ell+s_n}{L_N}\right) e_{\ell\tetra,\eta}(\zeta)\\
- C\frac{\ell+s_n}{L_N} \dashint_{\RR^3} \zeta^{4/3}
- \frac{C}{\ell} \frac{|(L_N+s_N)\tetra|}{|L_N\tetra|}\dashint_{\RR^3}(\zeta+\zeta^2),
\end{multline*}
where we used estimates analogous to \cref{volbubound,volbdbound}. Taking the liminf as $N\to\infty$ and the limsup as $\ell\to\infty$, we arrive at
$$
\liminf_{N\to\infty} \wt{e}_{L_N\tetra,s_N}(\zeta)\ge \limsup_{\ell\to\infty} e_{\ell\tetra,\eta}(\zeta)= e_\NUEG(\zeta).
$$
Combining this with \cref{etranscomp}, we conclude that the two limits exist and equal.
\end{proof}

\subsection{Thermodynamic limit for general domains}
In this section, we prove \cref{thermolimthm}.  The proof is standard and the technique goes back to Fisher \cite{fisher1964free}.

\vspace{.3em}
\noindent\textbf{Proof of (i).} 
Let $\{\Omega_N\}\subset\RR^3$ and $\{\delta_N\}\subset\RR_+$ be sequences such that $\delta_N/|\Omega_N|^{1/3}\to 0$ and $\delta_N|\Omega_N|^{1/3}\to\infty$
and that the Fisher regularity condition $|\partial\Omega_N + B_r|\le Cr|\Omega_N|^{2/3}$ holds true for all $r\le|\Omega_N|^{1/3}/C$. 
It is straightforward to see that the monotonicity bound \cref{tetracomp} holds true with the large tetrahedron $\ell'\tetra$ replaced by a general domain $\Omega_N$,
\begin{multline*}
e_{\Omega_N,\eta_{\delta_N}}(\zeta)\ge \Bigg( 1-C\sigma \frac{\delta}{\ell} -C\sigma\frac{\ell+\delta+\delta_N}{|\Omega_N|^{1/3}}\Bigg) e_{\ell\tetra,\eta_{\delta}}(\zeta)
- C \frac{\ell+\delta+\delta_N}{|\Omega_N|^{1/3}} \dashint_{\RR^3} \zeta^{4/3}\\
-\frac{C}{\ell}\Bigg( (1+\hbar^2\delta^{-1}) \dashint_{\RR^3} \zeta + \delta^3\dashint_{\RR^3} \zeta^2 \Bigg).
\end{multline*}
Taking $\ell=|\Omega_N|^{1/6}$ and $\delta$ fixed, we find
$$
\liminf_{N\to\infty} e_{\Omega_N,\eta_{\delta_N}}(\zeta)\ge e_\NUEG(\zeta),
$$
where the r.h.s. is given by the limit for tetrahedra (see \cref{tetrathermo}).

The proof of the upper bound is similar to that of \cref{tetrathermo} (ii), but instead of the big tetrahedron $\ell'\tetra$, we use the general domain $\Omega_N$. 
We obtain
\begin{multline*}
e_{\Omega_N,\eta_{\delta_N}}(\zeta)\le \left(1+C\sigma\frac{\ell+\delta+\delta_N}{|\Omega_N|^{1/3}}\right)\dashint_{1-\alpha}^{1+\alpha}\frac{\dd t}{t^4} e_{t\ell\tetra,\eta_{t\delta}}(\zeta)\\
+C\hbar^2\frac{\ell+\delta+\delta_N}{|\Omega_N|^{1/3}}\Bigg[ \dashint_{\RR^3} \zeta^{5/3} + \frac{1}{\ell\delta}\dashint_{\RR^3}\zeta + \dashint_{\RR^3} |\grad\sqrt{\zeta}|^2\Bigg] 
+ \frac{C\hbar^2}{\delta_N|\Omega_N|^{1/3}}\dashint_{\RR^3}\zeta\\
+C\delta^2\log(\alpha^{-1}) \dashint_{\RR^3}\zeta^2.
\end{multline*}
Choosing $\ell=|\Omega_N|^{1/6}$ and $\delta=|\Omega_N|^{-1/12}$, applying \cref{monoavg} on the $t$-averaged energy and taking the limsup as $N\to\infty$, we obtain
$$
\limsup_{N\to\infty} e_{\Omega_N,\eta_{\delta_N}}(\zeta)\le e_\NUEG(\zeta).
$$

\vspace{.3em}
\noindent\textbf{Proof of (ii).} Similarly to the lower bound in the proof of \cref{tetraeq} above, we can deduce 
$$
\liminf_{N\to\infty} \wt{e}_{\Omega_N,s_N}(\zeta)\ge e_\NUEG(\zeta).
$$

Next we show the upper bound. Fix $N\in\NN$ large, $a\in\RR^3$ and $Q\in\SO(3)$.
Then by setting 
$$
\rho_N=(\iden_{\Omega_N}*\eta_\delta)\zeta(Q(\cdot-a))
$$
we have
$\zeta(Q(\cdot-a))\iden_{\Omega_N^{s_N-}}\le \rho_N \le \zeta(Q(\cdot-a))\iden_{\Omega_N^{s_N+}}$. Hence,
$$
\inf_{\substack{\sqrt{\rho}\in H^1(\RR^3)\\ \zeta(Q(\cdot-a))\iden_{\Omega_N^{s_N-}}\le \rho \le \zeta(Q(\cdot-a))\iden_{\Omega_N^{s_N+}} }} \frac{E(\rho)}{|\Omega_N|} \le \frac{E(\rho_N)}{|\Omega_N|}.
$$
Clearly, we have $\rho_N\equiv \zeta(Q(\cdot-a))$ on $\Omega_N^{s_N-}$.
Next, using \cref{impdecup}
\begin{align*}
&\wt{e}_{\Omega_N,s_N}(\zeta)\le \dashint_{\RR^3}\dd a\int_{\SO(3)}\dd Q\, \frac{E(\rho_N)}{|\Omega_N|} \\
&\le \frac{1}{|\Omega_N|}\dashint_{\RR^3}\dd a\int_{\SO(3)}\dd Q\dashint_{1-\alpha}^{1+\alpha}\frac{\dd t}{t^4} \int_{\SO(3)}\dd R \dashint_{C_{t(\ell+\delta)}}\dd\tau\\
&\quad\quad\times \sum_{z\in\ZZ^3}  \sum_{j=1}^{24} E(\iden_{t(\ell+\delta) T_j\beta\tetra} * \eta_\delta(R\cdot - t(\ell+\delta) z - \tau)\rho_N)\\
&+\frac{1}{|\Omega_N|} \dashint_{\RR^3}\dd a\int_{\SO(3)}\dd Q \Biggl[ \frac{C\hbar^2\delta}{\ell}\int_{\RR^3}\rho_N^{5/3} +  \frac{C\hbar^2}{\ell\delta} \int_{\RR^3} \rho_N + \frac{C\hbar^2\delta}{\ell}\int_{\RR^3} |\grad\sqrt{\rho_N}|^2\\
&\quad\quad+C\delta^2\log(\alpha^{-1})\int_{\RR^3} \rho_N^2 \Biggr]\\
&=: \mathrm{(I)} + \mathrm{(II)}
\end{align*}
where we used the notation $\beta=\ell/(\ell+\delta)$ as above. 
It is clear that
\begin{align*}
\frac{1}{|\Omega_N|}\dashint_{\RR^3}\dd a\int_{\SO(3)}\dd Q\int_{\RR^3}\rho_N^{\alpha}&\le \dashint_{\RR^3} \zeta^\alpha
\end{align*}
and also
\begin{align*}
&\frac{1}{|\Omega_N|}\dashint_{\RR^3}\dd a\int_{\SO(3)}\dd Q\int_{\RR^3}|\grad\sqrt{\rho_N}|^2
\le \frac{C}{|\Omega_N|^{1/3}\delta}\dashint_{\RR^3}\zeta + C \dashint_{\RR^3} |\grad\sqrt{\zeta}|^2
\end{align*}
using \cref{indgradlemma}. 

Define the index set
$$
J:=\left\{ \begin{aligned}(z,j)\in \ZZ^3\times\{1,\ldots,24\} : &\supp( \iden_{t(\ell+\delta) T_j\beta\tetra} * \eta_\delta(R\cdot - t(\ell+\delta) z - \tau)) \cap \Omega_N^{s_N+}\neq \emptyset,\\
&\text{for some}\;R\in\SO(3)\;\text{and}\; \tau\in C_{t(\ell+\delta)} \end{aligned} \right\},
$$
which collects all the small smeared tetrahedra that possibly intersect with the outer domain $\Omega_N^{s+}$. Furthermore, define
$$
J_0:=\left\{ \begin{aligned}(z,j)\in \ZZ^3\times\{1,\ldots,24\} : &\supp(\iden_{t(\ell+\delta) T_j\beta\tetra} * \eta_\delta(R\cdot - t\ell z - \tau)) \subset \Omega_N^{s_N-},\\
&\text{for all}\;R\in\SO(3)\;\text{and}\; \tau\in C_{t(\ell+\delta)} \end{aligned} \right\},
$$
which contains small smeared tetrahedra that lie inside the inner domain $\Omega_N^{s-}$.
The $J_0$-part of (I) reads
\begin{multline*}
\mathrm{(Ia)}=\frac{1}{|\Omega_N|}\dashint_{\RR^3}\dd a\int_{\SO(3)}\dd Q \dashint_{1-\alpha}^{1+\alpha}\frac{\dd t}{t^4} \int_{\SO(3)}\dd R \dashint_{C_{t(\ell+\delta)}}\dd\tau \\
\times \sum_{(z,j)\in J_0} E\bigl(\iden_{t(\ell+\delta) T_j\beta\tetra} * \eta_\delta(R\cdot - t(\ell+\delta) z - \tau)\zeta(Q(\cdot-a))\bigr),
\end{multline*}
which gives similarly as above,
$$
\frac{|J_0||\ell\tetra|}{|\Omega_N|} \dashint_{1-\alpha}^{1+\alpha}\frac{\dd t}{t^4} e_{t\ell\tetra,\eta_\delta}(\zeta).
$$
The $(J\setminus J_0)$-part of (I) is bounded similarly as \cref{ebduest}, and we find
$$
\mathrm{(Ib)}\le C\hbar^2 \frac{|J\setminus J_0||\tetra|}{|\Omega_N|}\Bigg[ \dashint_{\RR^3} \zeta^{5/3} + \frac{1}{\ell\delta}\dashint_{\RR^3}\zeta + \dashint_{\RR^3} |\grad\sqrt{\zeta}|^2\Bigg] + \frac{C\hbar^2}{|\Omega_N|^{1/3}\delta}\dashint_{\RR^3}\zeta.
$$
The set $J \setminus J_0$ contains tetrahedra at a distance $\OC(\ell + \delta + s_N)$ of $\partial\Omega_N$. Since $|\partial\Omega_N|\le C |\Omega_N|^{2/3}$ by Fisher regularity, the volume fraction is $\OC((\ell + \delta + s_N)/|\Omega_N|^{1/3}) \to 0$ as $N\to\infty$. 
Applying \cref{monoavg} on the $t$-averaged energy of (I), we obtain after taking the first the limsup as $N\to\infty$ and then $\delta/\ell\to 0$, $\ell\delta\to\infty$,
$$
\limsup_{N\to\infty} \wt{e}_{\Omega_N,s_N}(\zeta)\le e_\NUEG(\zeta).
$$
This completes the proof of \cref{thermolimthm}.

\subsection{Local density approximation}
In this section, we briefly sketch the proof of \cref{qlda}. It is more complicated than
the proof of the classical case, which we detailed in \cref{clldaproof}. Here,
we do not attempt to produce explicit constants.

Analogously to the classical case, our bound follows from \emph{a priori} estimates
of \cref{qnuegrmrk} (ii)--(iii) for large $\epsilon$. Hence, we may assume that $\epsilon$ is sufficiently small. 

To show the upper bound, we set $\ell=\epsilon^{-3/2}$ and $\delta=\sqrt{\epsilon}$.
 and start with \cite[Eq. (75)]{lewin2020local}, which implies
\begin{multline*}
E((\iden_{\ell\tetra}*\eta_\delta)\zeta)\le \int_{\RR^3} (\iden_{\ell\tetra}*\eta_\delta)(x) e_\UEG(\zeta(x))\,\dd x
+ C\epsilon\int_{\RR^3} (\iden_{\ell\tetra}*\eta_\delta)\bigl(\zeta + \zeta^2\bigr)\\
+\frac{C}{\epsilon^{4p-1}} \int_{\ell\tetra + B_\delta} |\grad\zeta^\theta|^p + \frac{C}{\epsilon}\int_{\RR^3} (\iden_{\ell\tetra}*\eta_\delta)|\grad\sqrt{\zeta}|^2 
+C\int_{\RR^3} \zeta \left|\grad\sqrt{\iden_{\ell\tetra}*\eta_\delta}\right|^2
\end{multline*}
Replacing $\zeta$ by $\zeta(R(\cdot-a))$ and averaging over all translations and rotations, we find after dividing by the volume
$$
e_{\ell\tetra,\eta_\delta}(\zeta)\le  \dashint_{\RR^3} e_\UEG(\zeta(x))\,\dd x + C\epsilon\dashint_{\RR^3} \bigl(\zeta + \zeta^2\bigr) + 
\frac{C}{\epsilon^{4p-1}} \dashint_{\RR^3} |\grad\zeta^\theta|^p
+ \frac{C}{\epsilon}\dashint_{\RR^3} |\grad\sqrt{\zeta}|^2,
$$
which is our stated upper bound.

The lower bound is more complicated. Replacing $\ell\to t\ell$, $\delta\to t\delta$ and averaging in \cref{lboundmv} of \cref{lboundthm} we obtain for a large smeared tetrahedron of scale $\ell'\gg\ell$
\begin{multline*}
E((\iden_{\ell'\tetra}*\eta_{\delta'})\zeta)\ge \left(1-C\epsilon\right) \dashint_{1/2}^{3/2}\frac{\dd t}{t^4}\sum_{z\in\ZZ^3} \sum_{j=1}^{24} E\bigl(\iden_{t\ell\tetra_j} * \eta_{t\delta}(\cdot - t\ell z)(\iden_{\ell'\tetra}*\eta_{\delta'})\zeta\bigr)\\
 -C\epsilon \int_{\RR^3} (\iden_{\ell'\tetra}*\eta_{\delta'})(\zeta + \epsilon^2\zeta^2 )
\end{multline*}
As before, we split the sum in the first term according to $J_0$ and $J\setminus J_0$,
where $J_0$ and $J$ are defined as in the proof of \cref{monolemm}. The contribution of the $J_0$-part is
\begin{equation}\label{qldainte}
(1-C\epsilon)\dashint_{1/2}^{3/2}\frac{\dd t}{t^4} \sum_{(z,j)\in J_0}E\bigl(\iden_{t\ell\tetra_j} * \eta_{t\delta}(\cdot-t\ell z)\zeta\bigr).
\end{equation}
We claim that if for any $\zeta$ the bound
\begin{multline}\label{qldalow}
\dashint_{1/2}^{3/2}\frac{\dd t}{t^4}\Biggl[ E\bigl(\iden_{t\ell\tetra_j} * \eta_{t\delta}(\cdot-t\ell z)\zeta) - \int_{\RR^3} \iden_{t\ell\tetra_j}*\eta_{t\delta}(x-t\ell z)e_\UEG(\zeta(x))\, \dd x\\
+C\int_{\RR^3} \left|\grad\sqrt{\iden_{t\ell\tetra_j}*\eta_{t\delta}(\cdot-t\ell z)}\right|^2 \zeta 
+ C\epsilon \int_{\RR^3} \iden_{t\ell\tetra_j}*\eta_{t\delta}(\cdot-t\ell z)(\zeta+\zeta^2) \\
+\frac{C}{\epsilon} \int_{\RR^3} \iden_{t\ell\tetra_j}*\eta_{t\delta}(\cdot-t\ell z) |\grad\sqrt{\zeta}|^2 \Biggr] \ge -\frac{C}{\epsilon^{4p-1}} \int_{(2\ell\tetra_j+\ell z)+B_{2\delta}} |\grad\zeta^\theta|^p
\end{multline}
holds true for all $(z,j)\in J_0$ with a universal constant $C>0$, then the proof is finished. In fact, by replacing $\zeta$ by $\zeta(Q(\cdot-a))$ in \cref{qldalow} and averaging over all translations and rotations,
and plugging into \cref{qldainte} we find the stated lower bound in the limit $\delta'/\ell'\to 0$, $\ell'\delta'\to 0$. The boundary contributions in the $J\setminus J_0$-part result in an error term $-C\epsilon \dashint_{\RR^3} (\zeta + \zeta^2 )$ using the Lieb--Oxford inequality. 
Also, the term corresponding to $-C\epsilon$ in \cref{qldainte} is bounded using \cref{llsebound} and merged into the other terms. 

To show \cref{qldalow} we distinguish two cases just like in the classical case.

\noindent\textbf{Case I. (Simple tetrahedra)} If $(z,j)\in J_0$ is such that
\begin{multline}\label{simpineq}
\int_{\RR^3} \iden_{t\ell\tetra_j}*\eta_{t\delta}(\cdot-t\ell z)(\zeta^{4/3}+\zeta^{5/3})\\\le C\epsilon\int_{\RR^3} \iden_{t\ell\tetra_j}*\eta_{t\delta}(\cdot-t\ell z)(\zeta+\zeta^2) + \frac{1}{\epsilon^{4p-1}} \int_{(2\ell\tetra_j+\ell z)+B_{2\delta}} |\grad\zeta^\theta|^p,
\end{multline}
then using the \emph{a priori} bound
\begin{multline*}
\Biggl| E\bigl(\iden_{t\ell\tetra_j} * \eta_{t\delta}(\cdot-t\ell z)\zeta) - \int_{\RR^3} \iden_{t\ell\tetra_j}*\eta_{t\delta}(x-t\ell z)e_\UEG(\zeta(x))\, \dd x \Biggr|\\
\le C\int_{\RR^3} \iden_{t\ell\tetra_j}*\eta_{t\delta}(\cdot-t\ell z)(\zeta^{4/3}+\zeta^{5/3}) + C\int_{\RR^3} \iden_{t\ell\tetra_j}*\eta_{t\delta}(\cdot-t\ell z) |\grad\sqrt{\zeta}|^2\\
+C\int_{\RR^3} \left|\grad\sqrt{\iden_{t\ell\tetra_j}*\eta_{t\delta}(\cdot-t\ell z)}\right|^2 \zeta,
\end{multline*}
which in turn follows from \cref{llsbound} and \cref{lobound}, \emph{implies} 
\cref{qldalow}.

\noindent\textbf{Case II. (Main tetrahedra)} If the opposite of \cref{simpineq} holds in 
a tetrahedron $(z,j)\in J_0$, then we proceed analogously to \cref{clldaproof}.
We omit the details of the rest of the argument for brevity.

\subsection{Semiclassical bound} \label{subsec:Semi-limit}
Here, we prove the bound of \cref{semiclassthm}. Obviously, we have
$E^\hbar(\rho)\ge E^\cl(\rho)$ for all $\sqrt{\rho}\in H^1(\RR^3)$. The issue is that when we put $\rho=(\iden_{\ell'\tetra}*\eta_{\delta})\zeta$, the classical NUEG energy will involve a smooth cutoff instead of the hard one used to our definition. But this smooth cutoff will 
introduce an error which is negligible in the thermodynamic limit.
Using the Graf--Schenker inequality,
\begin{align*}
&e_{\ell'\tetra,\eta_{\delta'}}^\hbar(\zeta)
\ge \int_{\SO(3)}\dd Q\dashint_{\RR^3}\dd a\,\frac{E^\cl\bigl((\iden_{\ell'\tetra}*\eta_{\delta'})\zeta(Q(\cdot-a))\bigr)}{|\ell'\tetra|} \\
&\ge \frac{1}{|\ell'\tetra|} \int_{\SO(3)}\dd Q\dashint_{\RR^3}\dd a \int_{\SO(3)}\dd R \dashint_{C_\ell}\dd\tau \sum_{z\in\ZZ^3}\sum_{j=1}^{24} 
E^\cl\bigl((\iden_{\ell'\tetra}*\eta_{\delta'})\iden_{\ell\tetra_j}(R\cdot-\ell z-\tau)\zeta(Q(\cdot-a))\bigr) \\
&\quad -\frac{C}{\ell} \frac{1}{|\ell'\tetra|} \int_{\SO(3)}\dd Q\dashint_{\RR^3}\dd a \int_{\RR^3} (\iden_{\ell'\tetra}*\eta_{\delta'})\zeta(Q(\cdot-a))=:\mathrm{(I)} + \mathrm{(II)}.
\end{align*}
Here, (II) is simply $-\frac{C}{\ell} \dashint_{\RR^3}\zeta$. To deal with (I), we define the index set
$$
J:=\left\{ \begin{aligned}(z,j)\in \ZZ^3\times\{1,\ldots,24\} : &\supp( \iden_{\ell\tetra_j}(R\cdot - \ell z - \tau)) \cap \supp(\iden_{\ell'\tetra}*\eta_{\delta'})\neq \emptyset,\\
&\text{for some}\;R\in\SO(3)\;\text{and}\; \tau\in C_\ell \end{aligned} \right\},
$$
which collects all the small (sharp) tetrahedra $R^\top(\ell\tetra_j + \ell z + \tau)$ that possibly intersect with the big smeared tetrahedron $\supp(\iden_{\ell'\tetra}*\eta_{\delta'})$. 
Furthermore, define
$$
J_0:=\left\{ \begin{aligned}(z,j)\in \ZZ^3\times\{1,\ldots,24\} : &\supp( \iden_{\ell\tetra_j}(R\cdot - \ell z - \tau)) \subset (\ell'-\delta')\tetra,\\
&\text{for all}\;R\in\SO(3)\;\text{and}\; \tau\in C_\ell \end{aligned} \right\},
$$
which contains the small tetrahedra that are well inside $\ell'\tetra$. Using this, the term (I) can be split into two parts. The first of which is
\begin{align*}
\mathrm{(Ia)}&= \frac{1}{|\ell'\tetra|} \int_{\SO(3)}\dd Q\dashint_{\RR^3}\dd a \int_{\SO(3)}\dd R \dashint_{C_\ell}\dd\tau \sum_{(z,j)\in J_0} 
E^\cl\bigl(\iden_{\ell\tetra_j}(R\cdot-\ell z-\tau)\zeta(Q(\cdot-a))\bigr) \\
&=\frac{|J_0|}{|\ell'\tetra|} \int_{\SO(3)}\dd Q\dashint_{\RR^3}\dd a\,
E^\cl\bigl(\iden_{\ell\tetra}\zeta(Q(\cdot-a))\bigr) \\
&=\frac{|J_0||\ell\tetra|}{|\ell'\tetra|} e_{\ell\tetra}^\cl(\zeta).
\end{align*}
The second part is
\begin{align*}
\mathrm{(Ib)}&= \frac{1}{|\ell'\tetra|} \int_{\SO(3)}\dd Q\dashint_{\RR^3}\dd a \int_{\SO(3)}\dd R \dashint_{C_\ell}\dd\tau \sum_{(z,j)\in J\setminus J_0} 
E^\cl\bigl(\iden_{\ell\tetra_j}(R\cdot-\ell z-\tau)\zeta(Q(\cdot-a))\bigr) \\
&\ge - \frac{|J\setminus J_0|}{|\ell'\tetra|} \int_{\SO(3)}\dd Q\dashint_{\RR^3}\dd a\,
c_\LO \int_{\RR^3} \iden_{\ell\tetra}\zeta(Q(\cdot-a))^{4/3} \\
&\ge - \frac{|J\setminus J_0||\ell\tetra|}{|\ell'\tetra|} c_\LO \dashint_{\RR^3} \zeta^{4/3},
\end{align*}
using the Lieb--Oxford inequality. The volume ratios obey similar bounds as in the proof of \cref{monolemm}, hence by
collecting the above estimates and taking the thermodynamic limit $\delta'/\ell'\to 0$, $\ell'\delta'\to\infty$ and $\ell\to\infty$, we obtain
$e_\NUEG^\hbar(\zeta)\ge e_\NUEG^\cl(\zeta)$,
and from this
$\liminf_{\mu\to 0} \wh{e}_\NUEG(\mu;\zeta_0)\ge e_\NUEG^\cl(\zeta_0)$ as stated.

\appendix

\section{Proof of \cref{permeanthm}}\label{furthersec}
Let first $a=0$ and consider the set of lattice vectors 
of tiles that lie completely inside $C_L$,
$$
\LC_0=\{ l\in\LC' : (\Lambda+l)\subset C_L \},
$$
where $\LC'=\{ l\in\LC : (\Lambda+l)\cap C_L\neq\emptyset\}$,
is the set of lattice vectors of tiles that intersect with $C_L$.
Then we have by the $\LC$-periodicity of $u$,
\begin{align*}
\dashint_{C_L} u&=\frac{1}{|C_L|} \sum_{l\in\LC_0} \int_{\Lambda+l} u + \frac{1}{|C_L|} \sum_{l\in \LC'\setminus\LC_0} \int_{(\Lambda+l)\cap C_L} u\\
&=\frac{|\LC_0||\Lambda|}{|C_L|} \dashint_{\Lambda} u + \frac{1}{|C_L|} \sum_{l\in \LC'\setminus\LC_0} \int_{(\Lambda+l)\cap C_L} u
\end{align*}
Noting that the cells corresponding to $\LC'\setminus\LC_0$ occupy a volume at most $C_\Lambda L^{d-1}|\Lambda|$ for some constant $C_\Lambda>0$ depending only on $\Lambda$, so
$$
\frac{|\LC_0||\Lambda|}{|C_L|}\ge 1 - \frac{C_\Lambda|\Lambda|}{L}.
$$
In summary, we obtain the estimate
\begin{align*}
\left|\dashint_{C_L} u-\dashint_{\Lambda} u\right|\le \left(1-\frac{|\LC_0||\Lambda|}{|C_L|} + \frac{|\LC'\setminus\LC_0||\Lambda|}{|C_L|}\right) \dashint_{\Lambda} |u|\le \frac{C_\Lambda|\Lambda|}{L} \dashint_{\Lambda} |u|
\end{align*}
which proves \cref{permean} for $a=0$. For the general case, we replace $u$ by $u(\cdot+a)$ and note that $\dashint_{\Lambda+a}u=\dashint_{\Lambda} u$ by periodicity. The proof is finished.

\section{Weak-* lower semicontinuity of $e_\NUEG^\hbar$}
This appendix is devoted to the proof of

\begin{theorem}\label{enueglsc}
The functional $\zeta\mapsto e_\NUEG^\hbar(\zeta)$ is weak-* lower semicontinuous in the following sense: for any sequence $\{\sqrt{\zeta_j}\}\subset H_\per^1(\LC)$ and $\sqrt{\zeta}\in H^1_\per(\LC)$ such that $\sqrt{\zeta_j}\wconv\sqrt{\zeta}$ in $\Hloc^1(\RR^3)$, we have
$$
e_\NUEG^\hbar(\zeta)\le \liminf_{j\to\infty} e_\NUEG^\hbar(\zeta_j).
$$
\end{theorem}

We set $\hbar=1$ and drop it from the notation for clarity. We first recall the following weak-* lower semicontinuity property of the grand-canonical Levy--Lieb functional from \cite{lewin2019universal}.
\begin{theorem}[Lewin--Lieb--Seiringer]\label{Flsc}
For any sequence $\{\sqrt{\rho_j}\}\subset H^1(\RR^3)$ such that $\sqrt{\rho_j}\wconv \sqrt{\rho}$ in $\dot{H}^1(\RR^3)$, there holds
$F_\LL(\rho)\le \liminf_{j\to\infty} F_\LL(\rho_j)$.
\end{theorem}
The crucial point here is that $\int_{\RR^3}\rho_j$ need not converge to $\int_{\RR^3}\rho$, in particular mass may be lost. 
Still, the weak convergence of the gradient of the ``wavefunction'' $\sqrt{\rho_j}$ is sufficient to establish the lower semicontinuity of $F_\LL(\rho)$.

As a preparation, we need a simple lemma.
\begin{lemma}\label{wconvlemm}
Suppose that $\{\sqrt{\zeta_j}\}\subset H^1_\per(\LC)$ and $\sqrt{\zeta}\in H^1_\per(\LC)$ such that 
$\sqrt{\zeta_j}\wconv\sqrt{\zeta}$ in $\Hloc^1(\RR^3)$.
Then $\sqrt{(\iden_\Omega*\eta_\delta)\zeta_j}\wconv \sqrt{(\iden_\Omega*\eta_\delta)\zeta}$
in $\dot{H}^1(\RR^3)$.
\end{lemma}
\begin{proof}
We have for for every $\phi\in \dot{H}^1(\RR^3)$,
\begin{align*}
&\int_{\RR^3} \left(\grad\sqrt{(\iden_\Omega*\eta_\delta)\zeta_j} - \grad\sqrt{(\iden_\Omega*\eta_\delta)\zeta}\right)\cdot\grad\phi\\
&=\int_{\RR^3} (\sqrt{\zeta_j}-\sqrt{\zeta})\grad\sqrt{\iden_\Omega*\eta_\delta}\cdot\grad\phi+
\int_{\RR^3} (\iden_\Omega*\eta_\delta)\left(\grad\sqrt{\zeta_j} -\grad\sqrt{\zeta}\right)\cdot\grad\phi,
\end{align*}
where in the first term, 
$$
\|\grad\sqrt{\iden_\Omega*\eta_\delta}\cdot\grad\phi\|_{L^2}\le \|\grad\sqrt{\iden_\Omega*\eta_\delta}\|_{\infty} \|\grad\phi\|_{L^2}<\infty.
$$
Since $\supp(\grad\sqrt{\iden_\Omega*\eta_\delta})$ is bounded, our hypothesis applies and we conclude that the r.h.s. tends to 0. This finishes the proof. 
\end{proof}

Next, we show the weak-* lower semicontinuity of the indirect energy per volume.

\begin{lemma}\label{evollsc}
Suppose that $\{\sqrt{\zeta_j}\}\subset H^1_\per(\LC)$ and $\sqrt{\zeta}\in H^1_\per(\LC)$ such that 
$\sqrt{\zeta_j}\wconv\sqrt{\zeta}$ in $\Hloc^1(\RR^3)$. Then
$$
e_{\ell\tetra,\eta_\delta}(\zeta)\le \liminf_{j\to\infty} e_{\ell\tetra,\eta_\delta}(\zeta_j).
$$
\end{lemma}
\begin{proof}
Instead of $\dashint_{\RR^3}\dd a$ we might as well write $\dashint_\Lambda\dd a$ for a unit cell $\Lambda$ of $\LC$ in the definition of $e_{\ell\tetra,\eta_\delta}(\zeta)$. By \cref{Flsc} and the continuity
of $D(\rho)$, we have
\begin{equation}\label{evolscpr}
\begin{aligned}
e_{\Omega,\eta_\delta}(\zeta)
&\le \int_{\SO(3)}\dd  R \dashint_{\Lambda}\dd a\, \liminf_{j\to\infty} \frac{E\bigl((\iden_{\ell\tetra}*\eta_\delta)\zeta_j(R(\cdot-a))\bigr)}{|\ell\tetra|}\\
&=\int_{\SO(3)}\dd  R \dashint_{\Lambda}\dd a\, \liminf_{j\to\infty} f_{\ell\tetra,\eta_\delta,\zeta_j(R\cdot)}(a),
\end{aligned}
\end{equation}
where $f_{\Omega,\eta_\delta,\zeta}(a)$ was defined in \cref{fdef}.
Here, the sequence of functions $f_{\ell\tetra,\eta_\delta,\zeta_j(R\cdot)}$ can be bounded from below by Lieb--Oxford inequality as
\begin{align*}
f_{\ell\tetra,\eta_\delta,\zeta_j(R\cdot)}(a)&\ge - C\frac{1}{|\ell\tetra|}\int_{\RR^3} (\iden_{\ell\tetra}*\eta_\delta)\zeta_j^{4/3}(R(\cdot-a)).
\end{align*}
Consider the set of lattice vectors of the rotated tiles that lie well inside $\iden_{\ell\tetra}*\eta_\delta$,
$$
\LC_0=\{ l\in\LC' : R^\top(\Lambda+l)\subset (\ell-\delta)\tetra \},
$$
where
$\LC'=\{ l\in\LC : R^\top(\Lambda+l)\cap \supp(\iden_{\ell\tetra}*\eta_\delta)\neq\emptyset \}$.
We have
\begin{align*}
&\frac{1}{|\ell\tetra|}\int_{\RR^3} (\iden_{\ell\tetra}*\eta_\delta)\zeta_j^{4/3}(R(\cdot-a))\\
&=\frac{1}{|\ell\tetra|}\sum_{l\in\LC_0}\int_{R^\top(\Lambda+l)}\zeta_j^{4/3}(R(\cdot-a))
+\frac{1}{|\ell\tetra|}\sum_{l\in\LC'\setminus\LC_0} \int_{R^\top(l+\Lambda)} (\iden_{\ell\tetra}*\eta_\delta)\zeta_j^{4/3}(R(\cdot-a))\\
&\le \frac{|\LC_0||\Lambda|}{|\ell\tetra|}\dashint_{\Lambda}\zeta_j^{4/3}
+\frac{1}{|\ell\tetra|}\sum_{l\in\LC'\setminus\LC_0} \int_{R^\top(l+\Lambda)}\zeta_j^{4/3}(R(\cdot-a))\\
&\le \dashint_{\Lambda}\zeta_j^{4/3}
+\frac{|\LC'\setminus\LC_0||\Lambda|}{|\ell\tetra|}\dashint_{\Lambda}\zeta_j^{4/3}.
\end{align*}
Here, $|\LC'\setminus\LC_0||\Lambda|\le C_{\Lambda} \delta \ell^2$, so we obtain the bound
\begin{align*}
f_{\Omega,\eta_\delta,\zeta_j(R\cdot)}(a)&\ge - \frac{C_\Lambda\delta}{\ell} \dashint_{\Lambda}\zeta_j^{4/3}.
\end{align*}
By the weak convergence of $\sqrt{\zeta_j}$, we have $\int_\Lambda\zeta_j\le C$ and $\int_\Lambda |\grad\sqrt{\zeta_j}|^2\le C$. The Sobolev-, and H\"older inequalities imply that $\int_\Lambda \zeta_j^{4/3}\le C$. In conclusion, 
$f_{\ell\tetra,\eta_\delta,\zeta_j(R\cdot)}(a)\ge -C$ for $a\in\Lambda$, so Fatou's lemma may be used to interchange the liminf and the integrations in \cref{evolscpr}, 
which finishes the proof.
\end{proof}

By combining the above finite volume weak-* lower semicontinuity property with the convergence rate estimates from \cref{tetrathermo}, we can complete the proof of \cref{enueglsc}.

\begin{proof}[Proof of \cref{enueglsc}]
Using  our convergence rate estimates \cref{tetrathermo} and \cref{evollsc}, we find 
\begin{align*}
e_\NUEG(\zeta)&\le e_{\ell\tetra,\eta_\delta}(\zeta) + \text{(error term of \cref{tetrathermo} (iii))}_{\ell,\delta}\\
&\le \liminf_{j\to\infty} e_{\ell\tetra,\eta_\delta}(\zeta_j) + \text{(error term of \cref{tetrathermo} (iii))}_{\ell,\delta}\\
&\le \liminf_{j\to\infty} \Big[ e_\NUEG(\zeta_j)  + \text{(error term of \cref{tetrathermo} (i))}_{\ell,\delta,j} \Big]\\
&\quad\quad + \text{(error term of \cref{tetrathermo} (iii))}_{\ell,\delta}\\
&\xrightarrow[\delta/\ell\to 0,\,\ell\delta\to\infty]{} \liminf_{j\to\infty} e_\NUEG(\zeta_j)
\end{align*}
where in the last step we used the fact that the error terms in the convergence rate bound are all uniformly bounded in $j$.
\end{proof}

\bibliographystyle{plainnat}
\bibliography{dft}
\end{document}